\documentclass{article}
\usepackage[utf8]{inputenc}
\usepackage[margin = 3.5cm]{geometry} 

\usepackage{authblk} 
\usepackage{graphicx} 
\usepackage{hyperref}
\usepackage{url}
\usepackage{amsmath}
\usepackage{amsfonts}
\usepackage{mathrsfs}
\usepackage{amsthm}
\usepackage{amssymb}
\usepackage{mathtools}
\usepackage{subcaption}
\usepackage[linesnumbered,ruled]{algorithm2e}
\usepackage{comment}
\usepackage{multirow}
\usepackage{accents}
\usepackage{tabularx}
\usepackage{booktabs}
\usepackage{cleveref}

\newtheorem{theorem}{Theorem}
\newtheorem{corollary}[theorem]{Corollary}
\newtheorem{lemma}[theorem]{Lemma}
\newtheorem{definition}[theorem]{Definition}
\newtheorem{remark}[theorem]{Remark}

\newcommand{\ILA}{\text{ILA}}
\newcommand{\ULA}{\text{ULA}}
\newcommand{\rand}{\text{RAND}}
\DeclareMathOperator*{\argmin}{arg\,min}

\usepackage{xcolor}

\newcommand{\regret}{\textsc{Regret}}
\newcommand{\val}{\textsc{Val}}

\title{Privacy, Prediction, and Allocation} 

\author[1]{Ben Jacobsen}
\author[2]{Nitin Kohli\thanks{NK was supported by a grant from the Gates Foundation.}}
\affil[1]{Department of Computer Sciences, University of Wisconsin-Madison}
\affil[2]{Center for Effective Global Action, UC Berkeley}

\date{}

\begin{document}

\maketitle


\begin{abstract}
Algorithmic predictions are increasingly used to inform the allocation of scarce resources. 
The promise of these methods is that, through machine learning, they can better identify the people who would benefit most from interventions. 
Recently, however, several works have called this assumption into question by demonstrating the existence of settings where simple, unit-level allocation strategies can meet or even exceed the performance of those based on individual-level targeting.
Separately, other works have objected to individual-level targeting on privacy grounds, leading to an unusual situation where a single solution, unit-level targeting, is recommended for reasons of both privacy and utility.
Motivated by the desire to fully understand the interplay of privacy and targeting levels, we initiate the study of aid allocation systems that satisfy \textit{differential privacy}, synthesizing existing works on private optimization with the economic models of aid allocation used in the non-private literature. 
To this end, we investigate private variants of both individual and unit-level allocation strategies in both stochastic and distribution-free settings under a range of constraints on data availability. 
Through this analysis, we provide clean, interpretable bounds characterizing the tradeoffs between privacy, efficiency, and targeting precision in allocation.
\end{abstract}

\section{Introduction}
\label{sec:intro}

Public institutions, humanitarian organizations, and private firms increasingly use data to design and deliver interventions across populations. 
Notable examples can be found across such varied domains as education~\cite{bruce2011track,perdomo_difficult_2023}, humanitarian response \cite{aiken2022machine, aiken2023program, kohli2024privacy, kohli_enabling_2024}, 
housing assistance~\cite{mashiat2025pays}, and financial inclusion~\cite{bjorkegren2020behavior,blumenstock2023big, kohli_enabling_2024}.
A common theme across these applications is the growing use of personal and behavioral data to decide how resources and benefits should be allocated to improve welfare outcomes.


When designing such interventions, the level at which aid is targeted has far-reaching policy implications~\cite{povertyactionlab2022targeting, worldbank2022revisiting}. Two choices are typical \cite{shirali_allocation_2024}:
individual-level allocation (ILA) draws on detailed data about specific people -- either directly observed or predicted -- to assign benefits one person at a time. It is often praised for precision and efficiency~\cite{kleinberg2018human} but criticized for its reliance on sensitive personal data~\cite{raviv2024citizens}. Unit-level allocation (ULA) instead uses aggregated data about groups, such as geographic areas or demographic clusters, to allocate resources collectively. It requires less personal information and may offer greater inclusivity, but this comes at the cost of coarser allocation precision \cite{aiken2024targeting}.
 
Privacy concerns arise at all levels of this process
~\cite{data_how_nodate,luvison_low-cost_2024, wang_not_2023}, but the focus of this work is on the privacy of the allocation itself. 
For example, when allocations are based on sensitive characteristics such as poverty status, the mere fact that an individual receives an anti‑poverty benefit discloses that characteristic, which has been shown to provoke shame or jealousy among peers
and to expose individuals to additional, well‑documented privacy risks 
\cite{kahn2025expanding}. A promising approach to address this 
problem is to use \textit{differential privacy} (DP), a rigorous mathematical framework for limiting the information that an algorithm leaks about its input. 
But although a number of prior works have investigated private variants of ILA from an optimization perspective~\cite{hsu_private_2016,hsu2016jointly,chen_noisy_nodate}, and there are examples of systems that have employed DP for both ULA \cite{kohli2024privacy} and ILA  \cite{kohli_enabling_2024} in specific applications, there is  limited understanding on the fundamental tradeoffs between privacy and efficiency in allocation systems more broadly. 
In particular, it is not clear to what extent the addition of privacy requirements should influence how administrators navigate the fundamental tension between spending their limited budget on collecting more data vs. helping more people, which has been much discussed in recent works~\cite{shirali_allocation_2024,perdomo_relative_nodate,fischer-abaigar_value_2025}.

In this paper, we help to close this gap by systematically analyzing the tradeoffs between ILA and ULA across different levels of data availability, subject to DP constraints. As a warm-up, we first consider an idealized setting with full access to each individual's welfare where the only challenge is to satisfy DP (\Cref{sec:direct}). We then turn our attention to more challenging settings where administrators have limited access to individual welfare scores and must balance paying for greater data access with allocating greater amounts of aid. 
In \Cref{sec:sampling}, we consider an administrator that can pay to sample welfare scores directly, and design strategies with robust worst-case bounds on regret even when those scores are fixed arbitrarily.
Then, in \Cref{sec:learning}, we introduce distributional assumptions by considering the setting (increasingly common in practice~\cite{aiken2022machine, aiken2023program,kohli_enabling_2024}) where an administrator has access to auxiliary information (features) on the whole population and samples individual welfare scores to train a predictive model. Building on sample complexity bounds for DP-SCO~\cite{feldman2020private}, we investigate when and how approaches based on statistical modeling can improve the efficiency of private allocations.

For each setting we consider, we derive clean, interpretable bounds describing the asymptotic regimes where one strategy outperforms the other (Sections \ref{sec:sampling-regimes} and \ref{sec:learning-regimes}). 
Perhaps surprisingly, we find that the choice of DP parameter has little impact on the optimal choice of strategy: we present simple, practical DP algorithms for ULA and ILA whose regret relative to non-private baselines vanishes asymptotically. 
Beyond the level of privacy, we find that ULA is more efficient than ILA when budgets are small, data is expensive, average welfare is low, inequality is high, and individual welfare is difficult to predict. 
These results are consistent with and extend those of other recent works~\cite{shirali_allocation_2024,perdomo_difficult_2023,perdomo_relative_nodate}, indicating that many of the defining features of ILA and ULA carry over when privacy constraints are introduced. 
Taken together, our results provide the first general framework for navigating tradeoffs between privacy, efficiency, and targeting precision in the design of systems for fair allocation.

\section{Preliminaries}
\label{sec:prelim}

\subsection{Differential Privacy}
\label{sec:dp}

\textit{Differential privacy} (DP) is a mathematical notion of privacy designed for statistical tasks~\cite{dwork2006calibrating}. Informally speaking, a randomized mechanism satisfies DP if changing one individual's data doesn't change the distribution of the mechanism's outputs by too much. Our privacy guarantees are stated in terms of zCDP, but can be easily translated to R\'enyi DP or approximate DP using standard results~\cite{bun2016concentrated}.

\begin{definition}[$\psi$-zero-Concentrated Differential Privacy \cite{bun2016concentrated}]\label{def:zCDP}
    A mechanism $M: \mathcal{X}^n \to \mathcal{Y}$ is $\psi$-zCDP
    if, for all $x, x' \in \mathcal{X}^n$ differing on a single entry and all $\alpha \in (1, \infty)$, we have
        $D_\alpha(M(x) \Vert M(x')) \leq \psi \alpha$,
    where $D_\alpha(M(x) \Vert M(x'))$ is the $\alpha$-R\'enyi divergence between the distribution of $M(x)$ and the distribution of $M(x')$.
\end{definition}

Definition \ref{def:zCDP} turns out to be slightly too strict for computing allocations~\cite{hsu_private_2016}: if the output of our algorithm is a full allocation, then satisfying DP implies that the aid allocated to an individual cannot depend too much on that same individual's welfare, which is clearly incompatible with effective targeting. It is therefore standard to use a slight relaxation of DP called Joint DP~\cite{kearns2015robustmediatorslargegames}. Essentially, instead of publishing our full allocation, we separately tell each individual whether or not they were chosen to receive aid. Satisfying Joint DP then requires that changing one individual's data won't change the joint distribution of \textit{everybody else's} outputs by too much, but the original individual's output can change arbitrarily. 


\begin{definition}[Joint $\psi$-zero-Concentrated Differential Privacy; Adapted from ~\cite{kearns2015robustmediatorslargegames}]
\label{def:Joint-zCDP}
    Let $M: \mathcal{X}^n \to \mathcal{Y}^n$ be a randomized mechanism that produces $n$ outputs $ y = (y_1,\ldots,y_n)$ when provided $x = (x_1,\ldots, x_n)$. Let $M(x)_{-i}$ denote the output of $M(x)$ without the $i^{th}$ component. We say that $M$ satisfies Joint $\psi$-zCDP if for all $i \in [n]$, all $x, x' \in \mathcal{X}^n$ differing in their $i^{th}$ entry, and all $\alpha \in (1, \infty)$, we have
    $D_\alpha(M(x)_{-i} \Vert M(x')_{-i}) \leq \psi \alpha$.
\end{definition}

A classic result from Joint DP is the \textit{Billboard Lemma}, which allows us to reduce the task of constructing Joint DP mechanisms to standard DP mechanisms.

\begin{lemma}[Billboard Lemma; Adapted from \cite{hsu_private_2016} Lemma 1]\label{lem:billboard} 
    Let $M: \mathcal{X}^n \to \mathcal{Y}$ satisfy $\psi$-zCDP. Then for any function $f: \mathcal{Y} \times \mathcal{X} \to \mathcal{Z}$, the mechanism $M': \mathcal{X}^n \to \mathcal{Z}^n$ defined by $M'(x)_i = f(M(x), x_i)$ satisfies Joint $\psi$-zCDP.
\end{lemma}

\subsection{Related Work}
\label{sec:related-work}

\subsubsection{Allocation Based on Prediction}

There is now a substantial body of work analyzing the promises and pitfalls of predictive systems for resource allocation.
Some papers emphasize the benefits of predictive models in improving the efficiency and targeting of limited resources~\cite{kleinberg2018human, aiken2023program,bruce2011track,aiken2024targeting,mashiat2025pays}. Others have studied the theoretical tradeoffs of prediction-guided allocation~\cite{casacuberta2026good,johnson_what_2022,shirali_allocation_2024,perdomo_relative_nodate} and empirically investigated the value of prediction in allocation across different contexts \cite{fischer-abaigar_value_2025, perdomo_difficult_2023,fischerabaigar2026empiricallyunderstandingvalueprediction}. 

A distinct line of works has criticized the concept of predictive optimization on normative grounds, arguing contra~\cite{kleinberg2018human} that such systems often fail on their own terms
~\cite{wang2024against,barabas2018interventions,liu2023reimagining}. In parallel, a growing body of research has suggested that many aspects of life trajectories may be fundamentally difficult to predict regardless of the method used~\cite{salganik_measuring_2020,lundberg_origins_2023,kim_is_2023}. 


Most relevant to our study is the framework proposed by Shirali et al.~\cite{shirali_allocation_2024}, who characterize the conditions necessary for ILA to outperform ULA. They develop a simple mathematical model for aid allocation and find that ILA only outperforms ULA when inequality across units is low and the budget is high. 
We extend their framework through a detailed examination of the impact of data constraints on efficiency, including constraints motivated by privacy.

\subsubsection{Differentially Private Allocation and Matching}

Joint DP (Definition \ref{def:Joint-zCDP}) was first formalized by \cite{kearns2015robustmediatorslargegames} in the context of mechanism design for mediated games, but it was quickly observed that the same definition could also be used to compute private matchings and allocations. The pioneering work in this direction is \cite{hsu_private_2016}, which studied a setting where $n$ agents each have private valuations over (subsets of) $k$ goods.
Subsequent works have extended these results to a wide range of matching and convex optimization problems under a variety of assumptions~\cite{hsu2016jointly,chen_noisy_nodate,dinitz2025differentially}.
The main feature that distinguishes these works from our own is that they all begin with the assumption that agents' valuations are known in advance, whereas we study tradeoffs between different strategies for learning those valuations.
Our work can therefore be seen as a bridge between existing theoretical works on Joint DP allocation and the growing non-private literature, described above, that examines the tradeoffs inherent to targeted aid allocation.


\subsection{Our Basic Model}
\label{sec:model}

Our work expands on the model of allocation presented in Shirali et al \cite{shirali_allocation_2024}. Our main goal in this section is to introduce the key terms and concepts of this model; the motivation and limitations of our modeling choices are discussed in greater detail in \Cref{sec:discussion}.

To start out, we assume that our population is divided $M$ disjoint units $U_1, \ldots, U_M$, each with population $N$, giving a total population size of $P=MN$. Each individual $i \in [P]$ has a welfare $w_i \in [0,1]$, and our goal is to maximize the sum of all welfare scores. To this end, we can pay a fixed cost of $c$ to allocate aid to an individual, subject to a budget constraint of $B$. 

We use $\tau_i$ to denote the \textit{treatment effect} of aiding individual $i$, defined as the marginal increase in their welfare after our intervention. Following \cite{shirali_allocation_2024}, we assume that $\tau_i = \min(1, w_i + \delta_w) - w_i$ for some fixed $\delta_w$, i.e. treatment effects are homogeneous and welfare is capped at 1. We say that an individual $i$ is ``low-welfare'' if $w_i \leq 1-\delta_w$ and ``high-welfare'' otherwise.

Deciding on an allocation can be seen as a matching problem between an administrator with $k = B/c$ indivisible, exchangeable goods and $P$ potential recipients, where the value of allocating aid to individual $i$ is $\tau_i$. 
We represent allocations with index sets, and define the \textit{value} of an allocation with index set $\mathcal{I}$ given a welfare vector $w$ as $\textsc{Val}(\mathcal{I}, w) \coloneq \sum_{i \in \mathcal{I}} \tau_i$. This quantity represents the total improvement in welfare that occurs because of our intervention. To evaluate a particular allocation, we define its \textit{regret} by how far it falls short of the best value we could have achieved:
\begin{equation}
    \regret(\mathcal{I}, w;k) = \max_{|\mathcal{J}|=k} \val(\mathcal{J}, w) - \val(\mathcal{I}, w)
\end{equation}
Abusing notation, we will also define the regret of an algorithm $A$ that outputs an allocation by $\regret(A, w;k) \coloneq \mathbb{E}_{\mathcal{I} \sim A(w)}[\regret(\mathcal{I}, w;k)]$, and likewise for $\val(A, w)$. The \textit{normalized regret} is defined as $\overline\regret (\cdot, w; k) \coloneq \frac{1}{\delta_w} \regret(\cdot, w;k)$. In some contexts, we will consider welfare scores drawn i.i.d. from some distribution $\mathcal{D}$, in which case the expectation is taken over the randomness of both our algorithm and the input.

When discussing specific algorithms, we annotate a strategy with a hat when it is derived from estimated welfare scores, as in $\widehat \ILA$, and with a tilde when it satisfies DP, as in $\widetilde \ULA$. The exact nature of the estimation or privacy protection should always be clear from context.

\subsubsection{Individual-Level Allocation (ILA)}

ILA aids individuals in ascending order by welfare. Let $s(i)$ denote the index of the individual with the $i$th lowest welfare. Then the value of the optimal non-private allocation is given by 
$\val(\ILA, w) \coloneq \sum_{i=1}^k \tau_{s(i)}$. This allocation is exactly optimal and so has zero regret by definition, but is also clearly not privacy-preserving. 

To remedy this, we design a variant of ILA that satisfies Joint DP (Definition \ref{def:Joint-zCDP}). We begin by sorting individuals in ascending order by either true welfare scores (if they are known) or by an estimate of the probability $\eta_i \coloneq \mathbb{P}[w_i > 1-\delta_w]$ (if $w$ must be estimated). We then learn a global decision rule $f: [0, 1] \to \lbrace 0, 1\rbrace$ subject to central DP and output a single bit $f(w_i)$ (resp. $f(\hat\eta_i)$) to each individual indicating whether they will receive aid.

To match the non-private baseline, we want $f$ to be a good approximation of the non-private decision rule $f(w) \approx \mathbb{I}[w \leq w_{s(k)}]$, taking care to handle situations where scores may be highly concentrated or tied. The main technical challenge here is to ensure that our global budget constraint is still satisfied with high probability, given that each individual allocation decision must be made locally. When analyzing the regret of private variants of ILA, we will make frequent use of the following lemma:
\begin{lemma}
\label{lem:ila-helper}
    Let $\hat \eta_i$ denote an estimate for the probability that $w_i \leq 1-\delta_w$, $\eta_i$ be the true probability, and let $\hat s(j)$ be the index of the person with the $j$th lowest estimated probability. Suppose that our allocation only aids $k' \leq k$ individuals based on ascending estimated probabilities of having low welfare. Then we have:
    \begin{equation}
    \overline\regret(\widehat \ILA, w;k) \leq (k-k') + k' \bar \eta_{k'} + \sum_{i \in \mathcal{I}} |\hat \eta_{i} - \eta_{i}|
    \end{equation}
    where $\mathcal{I} \coloneq \lbrace i = s(j) = \hat s(j') \mid (j\leq k' < j') \lor (j' \leq k' < j) \rbrace$ is the set of individuals wrongly included or excluded from our allocation because of estimation error, and $\bar \eta_{k'} \coloneq \frac{1}{k'} \sum_{i=1}^{k'} \eta_{s(i)}$ is the probability that a randomly selected individual among those with the $k'$ lowest conditional probabilities of having high welfare does, in fact, have high welfare.
\end{lemma}
\begin{remark}
    In \Cref{sec:direct} and \Cref{sec:sampling}, we treat welfare scores as fixed, in which case $\eta_i \in \lbrace 0,1 \rbrace$ and $\bar \eta_{k'}$ is typically equal to 0. We consider probabilistic welfare in \Cref{sec:learning}.
\end{remark}

\subsubsection{Unit-Level Allocation (ULA)}
\label{sec:ula}

ULA targets aid allocation to the worst-off units without distinguishing between individuals within the same unit. Specifically, let $\rho_j = \frac{1}{|U_j|}\sum_{i \in U_j} \mathbb{I}[w_i > 1-\delta_w]$ denote the probability that a randomly selected individual in unit $j$ has high welfare. With no privacy constraints, ULA greedily allocates aid to everybody in the units with the lowest $\rho_j$ scores until its budget is exhausted. To avoid complications from rounding, we permit ULA to allocate resources uniformly at random within the last chosen unit whenever $k$ is not evenly divisible by $N$. We denote the total number of units treated by $K = \lceil k/N\rceil$.

Letting $T_j$ denote the average treatment effect within unit $j$ and using $s(j)$ to denote the index of the unit with the $j$th lowest value in $\rho$, the value of the non-private allocation is:
\begin{equation}
    \val(\ULA, w) \coloneq \sum_{j=1}^{\lfloor k/N \rfloor} NT_{s(j)} + (k \text{ mod } N) T_{s(\lfloor k/N\rfloor+1)} \geq k \left( \frac{1}{K} \sum_{j=1}^{K} T_{s(j)} \right)
\end{equation}
In Lemma 4.1 of \cite{shirali_allocation_2024}, the authors show that $\sum_{j=1}^{K} T_{s(j)} \geq \sum_{j=1}^{K} \delta_w(1-\rho_{s(j)})$. This then implies that $\overline\regret(\ULA, w; k) \leq k \bar\rho_{K}$, where $\bar\rho_{K} \coloneq \frac{1}{K} \sum_{j=1}^{K} \rho_{s(j)}$. 

The private variant of ULA operates in the same way, except that it first computes $\widetilde\rho$, a $\psi$-zCDP estimate of vector of unit profiles $\rho$, and then allocates aid in ascending order based on this estimate. 
The following lemma allows us to decompose the worst-case regret of private ULA into the worst-case regret of non-private ULA plus the error from privacy:
\begin{lemma}\label{lem:ula-helper}
    Let $x_j^*$ denote the number of individuals in unit $j$ that receive aid under $\ULA(w)$, and let $\tilde x_j^*$ denote the same for $\widetilde\ULA(w)$. Then for any $p, q \in [1, \infty]$ such that $1/p + 1/q = 1$, we can upper-bound the worst-case normalized regret of the allocation based on $\widetilde\rho$ by:
    \begin{equation}
        \overline\regret(\widetilde \ULA, w;k)  \leq 
        k\bar\rho_K +
        \lVert x^* - \tilde x^*\rVert_p \lVert \widetilde\rho - \rho\rVert_q
    \end{equation}
    In particular, we can choose between any of the following bounds for this privacy term: 
    \begin{equation}
        \begin{cases}
           2\min(k,P-k) \max_j |\widetilde\rho_j - \rho_j| & p=1, q=\infty \\
           \sqrt{2\min(k,P-k)N \sum_{j=1}^M (\widetilde\rho_j - \rho_j)^2}   & p=2, q=2 \\
           N \sum_{j=1}^M |\widetilde\rho_j - \rho_j|  &p = \infty, q=1
        \end{cases}
    \end{equation}
\end{lemma}
\begin{remark}
    Lemma \ref{lem:ula-helper} generalizes the connection between $\val(\ULA, w)$ and error in unit-level statistics presented in~\cite[Appendix D.1]{shirali_allocation_2024}, which corresponds to the case of $p=\infty, q=1$.
\end{remark}

\subsubsection{Random Allocation}

Finally, we will also frequently consider the baseline strategy of allocating aid to individuals sampled uniformly without replacement from the population (RAND), which may be preferable to either ILA or ULA when privacy or budgetary constraints are too strict to permit any form of effective targeting.

\begin{lemma}\label{lem:random-helper}
    Define $\bar\rho_M = \frac{1}{M} \sum_{j=1}^M \rho_j$, and let $X$ be a hypergeometric random variable with population size $P$, number of samples $k$, and number of successes $P(1-\bar\rho_M)$. Then 
    $$ \overline\regret(\rand, w; k) \leq \mathbb{E}[k - X] = k\bar\rho_M$$
\end{lemma}

\subsubsection{The Role of Inequality}
\label{sec:gini}

Prior works \cite{shirali_allocation_2024,perdomo_difficult_2023} have identified inequality between units as a major factor determining the (in)efficiency of ULA. One way to understand the connection is to think of inequality between units as a measure of how \textit{informative} unit membership is; inequality is high exactly when unit membership is a useful proxy for welfare.
The following lemma formalizes this intuition by quantifying inequality in terms of the \textit{Gini coefficient} of the unit profile vector, which was also studied in~\cite{shirali_allocation_2024}:

\begin{lemma}\label{lem:gini-baseline}
    Let $\rho \in [0,1]^M$ be a unit profile vector with Gini coefficient $G$, defined as $G = \frac{1}{2M^2\bar\rho_M} \sum_{i, j} |\rho_i - \rho_j| = \frac{1}{M^2 \bar\rho_M} \sum_{1 \leq i \leq j \leq M} (\rho_{s(j)} - \rho_{s(i)})$. Then
        $\bar\rho_M - \bar\rho_K \geq \frac{GM\bar\rho_M}{M-1} > G\bar\rho_M$
\end{lemma}

Recall that the worst-case regret of ULA scales linearly with $\bar\rho_K$, while the worst-case regret of random allocation scales with $\bar\rho_M$. Lemma \ref{lem:gini-baseline} can therefore be interpreted as either a lower bound on the marginal value of ULA relative to \rand, or, by rearranging, as an upper-bound on the marginal regret of ULA relative to ILA.

\section{Private Allocation when Welfare is Directly Observable}\label{sec:direct}

As a warm-up, we assume that the true welfare scores of all individuals are already known to the administrator. In this simple setting, the only challenge is to ensure that the final allocation we select satisfies Joint DP without exceeding our budget constraints.

\subsection{ILA}

\begin{algorithm}
\KwIn{
\begin{itemize}
    \item Vector of welfare scores $w =[w_{1}, \ldots, w_{n}] \in [0,1]^n$
    \item Privacy parameter $\psi > 0$; bin width $\theta > 0$; noise parameter $s\geq 0$; aid budget $k$; \\failure probability $\beta$
\end{itemize}
}
\vspace{-2mm}
Sample noisy welfare scores $\hat w_i \sim \textsc{Unif}[w_i-s, w_i+s]$ for $1 \leq i \leq n$\\
Initialize interval $\mathscr{R} = [-s, 1+s]$ and define $B_1 = [\ell_1, r_1] = [-s, -s+\theta]$. \\
Recursively define bins $B_j = (\ell_j, r_j] = (r_{j-1}, r_{j-1}+\theta]$ for $2 \leq j \leq (1+2s)/\theta$. \\
Compute bin counts $C_j = \sum_{i=1}^n \mathbb{I}[\hat w_i \in B_j]$ \\
Let $\tilde{S}_1, \ldots, \tilde{S}_{(1+2s)/\theta} = \texttt{PrivatePartialSums}\left([ C_j ]_{j=1}^{(1+2s)/\theta}; ~\psi\right)$ (see Appendix \ref{app:private-partial_sums})\\
Let $j^* = \min \left\lbrace j ~ \Big| ~ \tilde S_j + \frac{1}{\sqrt{\psi}}
\left(1 + \frac{\log(\frac{1+2s}{\theta}) + \gamma_E}{\pi}\right)\left(\sqrt{\log\left(\frac{1+2s}{\theta}\right)} + \sqrt{\log\frac{2}{\beta}}\right) \geq k \right\rbrace$ \\
\textbf{Return} the triple $(\hat w_i, \ell_{j^*}, \mathbb{I}[\hat w_i \leq \ell_{j^*}])$ to each individual $i \in [n]$
\caption{Joint $\psi$-zCDP ILA}
\label{alg:ila-direct}
\end{algorithm}

Our strategy for privatizing ILA is presented in \autoref{alg:ila-direct}; we take as input $n$ welfare scores from the population to compute a private estimate of the empirical CDF (in this section, we set $n = P$; in other sections, $n \le P$). Using this, we derive a high-probability lower bound on $w_{s(k)}$, the maximum welfare treated by the optimal allocation. Our own allocation assigns aid to everybody whose (possibly noisy) welfare falls below this bound. 
The following results show that \autoref{alg:ila-direct} is private and provide high probability upper and lower bounds on the number of individuals treated:

\begin{lemma}\label{lem:ila-privacy}
    \autoref{alg:ila-direct} satisfies Joint $\psi$-zCDP for any valid $w, \psi, \theta, s,$ and $k$.
\end{lemma}

\begin{lemma}[\autoref{alg:ila-direct} does not exceed budget whp.]\label{lem:ila-direct-budget}
    Fix $w,\psi, \theta, s,$ and $k$, and let $\ell_{j^*}$ be the threshold returned by \autoref{alg:ila-direct}. Then with probability at least $1-\beta/2$, $\sum_{i=1}^n \mathbb{I}[\hat w_i \leq \ell_{j^*}] \leq k$.
\end{lemma}
\begin{theorem}[Conditions for \autoref{alg:ila-direct} to achieve low regret.]\label{thm:ila-direct-helper}
    Fix $\psi, \theta, s,$ and $k$, and let $\ell_{j^*}$ be the threshold returned by \autoref{alg:ila-direct}. Define the random variable $Y_{i,j} = \mathbb{I}[\hat w_i \in B_j]$. If (1) $\forall j$, $\lbrace Y_{i,j} \rbrace_{i=1}^P$ are jointly independent and (2) $\exists p \in [0,1]$ such that $\forall i,j$, $\mathbb{E}[Y_{i,j}] \leq p$, then with probability at least $1-\beta/2$, $k-\sum_{i=1}^P \mathbb{I}[\hat w_i \leq \ell_{j^*}]$ is at most:
    \begin{equation}
        ks + np +
    \left(\frac{2}{\sqrt{\psi}}
         \left(1 + \frac{\log(\frac{1+2s}{\theta}) + \gamma_E}{\pi}\right) \\
         + \sqrt{np}
    \right)\left(\sqrt{\log\frac{1+2s}{\theta}} + \sqrt{\log\frac{2}{\beta}}\right)
    \end{equation}
\end{theorem}

Combining the preceding two results, we have that with probability at least $1-\beta$, \autoref{alg:ila-direct} simultaneously stays under budget and satisfies our regret bound.
The next two corollaries combine Theorem \ref{thm:ila-direct-helper} and Lemma \ref{lem:ila-helper} to bound the regret of \autoref{alg:ila-direct} against both stochastic and oblivious adversaries when run on the full population:

\begin{corollary}[Stochastic adversary regret]\label{cor:ila-direct-iid}
    If $s=0$ and welfare scores are sampled i.i.d. from a distribution with maximum density $\gamma$, then the conditions of Theorem \ref{thm:ila-direct-helper} are satisfied with $p=\gamma\theta$. Under these assumptions, choosing $\theta=\frac{1}{P\pi\sqrt{\psi}}$ yields:
    \begin{equation}
        \overline\regret(\widetilde\ILA, w; k) \leq (1+o(1)) \cdot \left[ \frac{2\log^{3/2}P + \gamma + 2\log P\sqrt{\log(2/\beta)}}{\pi \sqrt{\psi}} \right]
    \end{equation}
\end{corollary}
\begin{corollary}[Oblivious adversary regret]\label{cor:ila-direct-adv}
    If $s>0$ and $w$ is arbitrary, then the conditions of Theorem \ref{thm:ila-direct-helper} are satisfied with $p=\frac{\theta}{2s}$. Choosing $s = \frac{1}{k\pi\sqrt{\psi}}$ and $\theta = \frac{2s\log^{3/2} P}{P\pi \sqrt{\psi}}$ yields: 
    \begin{align}
        \overline\regret(\widetilde\ILA, w; k) \leq (1+o(1)) \cdot \left[ \frac{3\log^{3/2}P + 2\log P\sqrt{\log(2/\beta)}}{\pi \sqrt{\psi}} \right]
    \end{align}
\end{corollary}

\subsection{ULA}

\begin{algorithm}
\KwIn{
\begin{itemize}
    \item Vector of welfare scores $w =[w_{1}, \ldots, w_{n}] \in [0,1]^n$
    \item Partition of $[n]$ into units $\mathcal{U} = \lbrace U_1, \ldots, U_M \rbrace$
    ; privacy parameter $\psi > 0$; aid budget $k$
\end{itemize}
}
\vspace{-2mm}
Sample $\widetilde \rho_j \sim \mathcal{N}(\rho_j, \frac{1}{2|U_j|\psi})$; 
Initialize $x_1, \ldots, x_M = 0$ \\
\While{$k > 0$} {
    Choose $j = \argmin \lbrace \widetilde \rho_j \mid U_j \in \mathcal{U} \rbrace$ \\
    Allocate $x_{j} = \min(|U_{j}|, k)$; \\
    Update budget $k = k - x_{j}$; Update remaining units $\mathcal{U} = \mathcal{U} \setminus  \lbrace U_{j} \rbrace$
}
\textbf{Return} The unit-level allocation $(x_1, \ldots, x_M)$
\caption{$\psi$-zCDP ULA with Public Unit Membership}
\label{alg:ula-direct}
\end{algorithm}

Our strategy for privatizing ULA is comparatively straightforward. The replace-one sensitivity of the unit profile vector $\rho$ is $1/N$, and so by Lemma \ref{lem:gaussian-mechanism}, adding Gaussian noise with variance $\sigma^2=(2\psi N^2)^{-1}$ suffices to satisfy $\psi$-zCDP. Denote the private unit profile vector by $\widetilde\rho$. Using Lemma \ref{lem:ula-helper} alongside standard results about the mean and concentration of various norms of Gaussian vectors, we derive the following unified bound (setting $n=P$ for this section):

\begin{theorem}\label{thm:ula-direct}
    Fix $\psi > 0$. Then, with probability at least $1-\beta$, we have that:
    \begin{equation}
        \overline\regret(\widetilde\ULA, w; k) \leq k\bar\rho_K + 
        \frac{1}{N\sqrt{\psi}} \cdot \min \begin{cases}
            2\min(k, P-k)\sqrt{\log(2M)} + \sqrt{\log(1/\beta)}\\
            \sqrt{\min(k, P-k)P} + \sqrt{\log(1/\beta)} \\
            \frac{P}{\sqrt{\pi}} + \sqrt{M\log(1/\beta)} 
        \end{cases}
    \end{equation}
\end{theorem}

\section{Accounting for the Cost of Sampling Welfare Scores}\label{sec:sampling}


We next consider a setting where the administrator must pay an upfront cost of $c_2n$ to observe the welfare of $n$ individuals, reflecting the cost of e.g. running an RCT or conducting a survey. This is in on top of the baseline cost of $ck'$ to allocate aid to $k'$ individuals, which introduces an exploration/exploitation-style tradeoff: we can improve the efficiency by collecting more data, but because our budget is fixed, this comes at the cost of aiding fewer people. 

\subsection{ILA}\label{sec:ila-sampling}

For the remainder of this section, we fix $k = B/c$ and define the ratio $\lambda = c_2/c$. Our goal is to select parameters $n,k'$ to minimize our regret relative to the best $k$-allocation, subject to the budget constraint that $ck' + c_2n = B \Rightarrow k' + \lambda n = k$. We derive matching upper and lower bounds for this problem, uncovering a sharp phase-transition: depending on how large $\lambda$ is, it is optimal to either allocate uniformly at random, or else sample at least $k/2$ people.

\begin{theorem}\label{thm:ila-sample-budget}
    There exists a strategy for ILA with sampling with worst-case expected regret
    \begin{equation}
        \overline\regret(\ILA_\lambda, w; k) \leq 
    \begin{cases}
        \frac{P\lambda k}{P\lambda+k} + \sqrt{\frac{\min(Pk, P^2\lambda)}{16(P\lambda+k)}} & \lambda < \frac{P-k}{P} \\
        k\left(1 - \frac{k}{P}\right) & \lambda \geq \frac{P-k}{P}
    \end{cases}
    \end{equation}

    Moreover, any strategy for ILA with sampling must suffer worst-case regret of at least:
    \begin{equation}
    \overline\regret(\ILA_\lambda, w; k) \geq
    \begin{cases}
        \frac{(P\lambda - \frac{1}{4}) k}{P\lambda + k} & \lambda < \frac{P-k}{P}\\
        k\left(1-\frac{k}{P}\right) & \lambda \geq \frac{P-k}{P}
    \end{cases}
\end{equation}
\end{theorem}


Combining Theorem \ref{thm:ila-sample-budget} and Corollary \ref{cor:ila-direct-adv} yields the following corollary, which summarizes the utility guarantees of private ILA in this setting:

\begin{corollary}
    Let $\widetilde \ILA_\lambda$ be the algorithm that samples $n=\frac{Pk}{P\lambda + k}$ scores and then uses \autoref{alg:ila-direct} to approximately allocate aid to the people with the $k' = \frac{k^2}{P\lambda+k}$ lowest scores in the sample. Then $M$ satisfies Joint $\psi$-zCDP and with probability at least $1-\beta$, $\overline\regret(\widetilde \ILA_\lambda, w; k)$ is at most:
    \begin{align*}
        \frac{P\left(\lambda k + \sqrt{\lambda k \log(3/\beta)} \right)}{P\lambda + k}
        + \sqrt{\frac{\min(Pk, P^2\lambda)}{16(P\lambda+k)}}
        + (1+o(1)) \left[\frac{3\log^{3/2}P + 2\log P\sqrt{\log(3/\beta)}}{\pi \sqrt{\psi}} \right]
    \end{align*}
\end{corollary}



\subsection{ULA}

We begin by conducting a stratified sample of $n$ individuals across units, giving us $n/M$ observations per unit which can be provided as input to \autoref{alg:ula-direct}. Our final error is a combination of the statistical error from sampling and the error from additive noise.

An interesting feature of this problem (also observed in \cite{casacuberta2026good}) is that the worst case in a statistical sense would occur if $\rho_j = 0.5$ for all $j$, but in this case, the excess regret of our \textit{ranking} would be 0. 
Building on this intuition, we show that for any fixed vector $\rho$, the average variance across all components must be small whenever inequality is high:

\begin{lemma}\label{lem:inequality-variance}
    Let $G$ be the Gini coefficient of the unit profile vector $\rho$. Then we have that 
        $\frac{1}{M}\sum_{j=1}^M \rho_j(1-\rho_j) \leq \bar\rho_M(1-\bar\rho_M) - 3\bar\rho_M^2 \cdot G^2$
\end{lemma}

Combining this result with Lemma \ref{lem:ula-helper} yields the following theorem:

\begin{theorem}\label{thm:ula-sample}
    Fix $k, \lambda$, and $\rho$. Then $\overline\regret(\widetilde\ULA_\lambda, w;k)$ is at most:
\begin{equation}
    \bar\rho_K k + (1-\bar\rho_K)\lambda n + \min\begin{cases}
        \sqrt{\frac{M^2}{2\psi n^2} + \frac{P-n}{4nN}} \cdot 2\min(k, P-k)\sqrt{2\log(2M)} \\
        \sqrt{\frac{M^2}{2\psi n^2} + \frac{(P-n)(\bar\rho_M(1-\bar\rho_M)-3\bar\rho_M^2G^2)}{nN}} \cdot \sqrt{2\min(k, P-k)P} \\
        \sqrt{\frac{M^2}{2\psi n^2} + \frac{(P-n)(\bar\rho_M(1-\bar\rho_M)-3\bar\rho_M^2G^2)}{nN}} \cdot P\sqrt{\frac{2}{\pi}}
    \end{cases}
\end{equation}

\begin{corollary}\label{cor:ula-sample-budget}
    Let $C = 
    \min\left(
        2\min(k, P-k)\sqrt{2\log(2M)},
        \sqrt{2\min(k, P-k)P},
        P\sqrt{\frac{2}{\pi}}
    \right)$. Choosing $n = \max \Big(\left(\frac{CM}{\lambda (2\psi)^{1/2}}\right)^{1/2}  , \left(\frac{C^2M}{16\lambda^2} \right)^{1/3}\Big)$, the marginal increase in the expected normalized regret of $\widetilde\ULA_\lambda$ over non-private ULA is at most:
     $2^{3/4} \left(CM\lambda/\sqrt{\psi}\right)^{1/2} + 2^{-4/3} \left(C^2M\lambda \right)^{1/3}$
\end{corollary}
    
\end{theorem}

\subsection{Asymptotic Regimes}\label{sec:sampling-regimes}

For both ULA and ILA, the price of privacy is asymptotically negligible provided that $k = \Theta(P)$. Understanding their asymptotic performance therefore boils down to a comparison of the linear terms in each strategy's regret bound. By combining our previous results, we can define three possible asymptotic regimes in terms of the key parameters $G, \lambda$, and $\bar\rho_M$:

\begin{itemize}
    \item \textbf{ULA $>$ ILA $>$ RAND}. This occurs whenever $\frac{k}{P} \cdot \frac{\bar\rho_M(1-G)}{1-\bar\rho_M(1-G)} < \lambda < 1-\frac{k}{P}$; sampling is cheap enough for ILA to be more effective than random allocation, but too expensive to justify the marginal cost of needing to sample a larger percentage of the population.
    \item \textbf{ULA $>$ ILA $=$ RAND}. This occurs whenever $1 - \frac{k}{P} \leq \lambda$ and $G > 0$; sampling is too expensive to meaningfully improve the efficiency of ILA, but ULA is still able to take advantage of the information conveyed by unit membership.
    \item \textbf{ILA $\geq$ ULA $\geq$ RAND}. This occurs whenever $\lambda < \min(1-\frac{k}{P}, \frac{k}{P} \cdot \frac{\bar\rho_M(1-G)}{1-\bar\rho_M(1-G)})$; sampling is sufficiently cheap for ILA to make efficient use of it, while low inequality and high average welfare make unit membership a weak signal for identifying the worst-off.
\end{itemize}

A particularly interesting consequence of the above is that ULA always dominates ILA when $\lambda \geq 1$, which can be interpreted as representing settings where the only efficient means to measure welfare (or treatment effects) is to actually allocate aid and observe the result.

\section{Learning Welfare Scores from Auxiliary Information}
\label{sec:learning}

In the final setting we consider, we can once again pay a cost of $c_2 = \lambda c$ to observe a single individual's welfare. But unlike in the previous setting, every individual also has a \textit{feature vector} $x_i \in \mathcal{X}$ and a binary label $y_i = \mathbb{I}[w_i > 1-\delta_w]$ drawn i.i.d. from a joint distribution $\mathcal{D}$ over $\mathcal{X} \times \lbrace 0,1 \rbrace$. 
We assume that features (but not labels) are already known to the administrator,
whose goal is to learn a convex model to predict $\eta \coloneq\mathbb{E}[y \mid x]$ subject to DP, with both labels and features considered private. For this purpose, we make use of Algorithm 2 of Feldman et al.~\cite{feldman2020private}, which achieves optimal rates on population risk (see Theorem \ref{thm:feldman}). 

Our final allocation is computed using only the data of individuals who weren't included in the original sample. This allows us to use parallel composition~\cite{mcsherry2009privacy}, and because both ULA and ILA sample $o(P)$ welfare scores, it induces only asymptotically negligible regret.

For the purposes of relating our model's risk to the downstream efficiency of the induced allocation, it is useful to additionally require that the loss functions we use are \textit{proper}~\cite{buja2005loss,bao2023proper,narasimhan2013relationship}, meaning that $\mathcal{L}(\theta)$ is minimized when $\hat\eta \coloneq f_\theta(x) = \eta$. 
In the sequel, we will focus primarily on the squared loss $f(\theta, x, y) = (y - f_\theta(x))^2$, a canonical proper loss function.

We define regret in this stochastic setting against the best allocation \textit{in hindsight}, i.e. $\mathbb{E}_{w \sim \mathcal{D}^P}[\max_{|\mathcal{J}| = k} \val(\mathcal{J}, w) ]$. As a consequence, ILA with a Bayes' optimal model could still suffer significant regret if the variance of $y|x$ is high. It will emerge that ULA is much more robust to this sort of unpredictability, which we formalize with the following definition:

\begin{definition}[Excess-Risk Decomposition]\label{def:excess-risk} Let $\ell$ denote the squared loss, and let $\mathcal{L}(\theta) \coloneq \mathbb{E}[\ell(\theta, x, y)] = \alpha$. Then, if $x,y \sim \mathcal{D}$,
\begin{align}
    \alpha = \mathbb{E}[(f_{\theta}(x) - y)^2] 
    &= \underbrace{\mathbb{E}[(f_{\theta}(x) - \mathbb{E}[y\mid x])^2]}_{E} + \underbrace{\mathbb{E}[(y - \mathbb{E}[y\mid x])^2]}_{\sigma^2}     \end{align}
Where $E$ is the excess risk, and $\sigma^2$ is the irreducible or Bayes' risk. The excess risk of the optimal classifier $\theta^* = \argmin_{\theta \in \Theta} \mathcal{L}(\theta)$ is denoted by $E^*$.
\end{definition}

\subsection{ILA}

ILA proceeds by first learning a model to predict the probability that an individual is low-welfare, and then allocating aid according to the predicted probabilities in ascending order. To start out, we take the expected loss of the learned model as a constant, and then later incorporate the error bounds from Theorem \ref{thm:feldman} to derive the optimal choice of $n$.

\begin{lemma}\label{lem:ila-sco}
    Let $f$ be a binary classifier with expected $\ell_2^2$ loss $\alpha$, and let $\widehat \ILA$ denote the algorithm that allocates aid to individuals in ascending order based on $\hat y_i = f(x_i)$. Then with probability at least $1-\beta$, we have:
    \begin{equation}
    \overline\regret(\widehat \ILA, w; k) \leq k' \cdot \bar \eta_{k'} + P\sqrt{\alpha} + P^{3/4}\sqrt{\log(1/\beta)/2}
\end{equation}
\end{lemma}

\begin{theorem}\label{thm:ila-sco-private}
    Let $\widetilde\ILA$ be the algorithm that samples $n$ individuals from the population uniformly at random, fits a model $\hat\theta$ with expected $\ell_2^2$ error $\alpha$ subject to $\psi$-zCDP, and then runs \autoref{alg:ila-direct} on the estimated scores of the remaining population with privacy parameter $\psi$. Then $\widetilde \ILA$ satisfies Joint $\psi$-zCDP, and:
    \begin{equation}
        \overline\regret(\widetilde \ILA, w; k) \leq n\lambda + k' \bar \eta_{k'} +P\sqrt{\alpha} + O\left( \frac{\log^{3/2} P}{\sqrt{\psi}} \right)
    \end{equation}
\end{theorem}

\begin{corollary}\label{cor:ila-sco-budget}
    When $\widetilde\ILA$ is instantiated with Algorithm 2 of \cite{feldman2020private}, its asymptotic normalized regret grows like:
    \begin{equation}
        n\lambda + k' \bar \eta_{k'} + P\sqrt{\sigma^2 + E^* + 10LD \left( \frac{1}{\sqrt{n}} + \frac{\sqrt{p}}{ n\sqrt{2\psi}} \right)} + O\left( \frac{\log^{3/2} P}{\sqrt{\psi}} \right)
    \end{equation}
    Choosing $n = \max\left( \left(\frac{P}{\lambda}\sqrt{10LD(2\psi)^{-1/2}\sqrt{p}}\right)^{2/3}, \left(\frac{P}{\lambda}\sqrt{10LD}\right)^{4/5} \right)$, we can guarantee a worst-case regret bound of:
\begin{equation}
    k' \bar \eta_{k'} +P\sqrt{\mathcal{L}(\theta^*)} + 2\lambda^{1/3} \left(\frac{P\sqrt{10LD\sqrt{p}}}{(2\psi)^{1/4}} \right)^{2/3} + 2\lambda^{1/5} \left( P\sqrt{10LD} \right)^{4/5}
\end{equation}
    If we additionally have access to a lower bound on the risk of the optimal classifier $\mathcal{L}(\theta^*) \geq L^\downarrow$, then we can instead choose $n = \max\left( \sqrt{\frac{5PLD\sqrt{p}}{\lambda  \sqrt{2\psi L^\downarrow}}}, \left(\frac{5PLD}{\lambda\sqrt{ L^\downarrow}}\right)^{2/3} \right)$ for an instance-specific regret bound of:
\begin{equation}
    k' \bar \eta_{k'} +P\sqrt{\mathcal{L}(\theta^*)} + 2\sqrt{\frac{5PLD\lambda\sqrt{p}}{\sqrt{2\psi L^\downarrow}}} + 2\left( \frac{5PLD\lambda}{\sqrt{L^\downarrow}} \right)^{2/3}
\end{equation}
\end{corollary}

\subsection{ULA}

In the previous settings we've considered, unit membership could be thought of as a particularly simple feature vector corresponding the standard basis of $\mathbb{R}^M$. We will generalize this idea by assuming we are given a partition $\mathscr{P}$ over the feature space $\mathcal{X} \subseteq \mathbb{R}^p$, which induces a partition over individuals into $|\mathscr{P}|$ units. This modeling decision is motivated by the sorts of bureaucratically legible and overlapping categories that are frequently used to target aid in existing systems~\cite{johnson_what_2022}, such as ``disabled veterans.'' We treat this choice of categories as exogenous and do not consider the problem of trying to \textit{learn} an  efficient partition.

One minor complication is that up until this point, we have assumed that the unit membership of each individual is publicly known,  
which may be inappropriate if units correspond to sensitive traits like poverty or disability rather than e.g. geographic regions. We therefore introduce \autoref{alg:ula-private-membership}, a Joint DP variant \autoref{alg:ula-direct} that can accommodate private unit membership at essentially no cost, provided that no unit is too small. 

\begin{algorithm}
\KwIn{
\begin{itemize}
    \item Vector of welfare scores $w =[w_{1}, \ldots, w_{P}] \in [0,1]^P$
    \item Partition of $[P]$ into units $\mathcal{U} = \lbrace U_1, \ldots, U_M \rbrace$ with minimum unit size $N$
    \item Privacy parameters $\psi_1 + \psi_2 =\psi$; aid budget $k$
\end{itemize}
}
\vspace{-2mm}

Sample $\widetilde \rho_j \sim \mathcal{N}(\rho_j, \frac{1}{N^2\psi_1})$; \\
For each individual $i \in [P]$ with corresponding unit $j$, set $\tilde w_i = \widetilde \rho_j$ \\
\textbf{Return} Output of \autoref{alg:ila-direct} with input $\tilde w, \psi_2$ and $s,\theta$ set as in Corollary \ref{cor:ila-direct-adv}

\caption{Joint $\psi$-zCDP ULA with Private Unit Membership}
\label{alg:ula-private-membership}
\end{algorithm}

\begin{lemma}\label{lem:private-membership}
    Choosing $\frac{\psi_2}{\psi_1} = \left( \frac{3\log^{3/2}P}{\pi} \cdot \frac{N}{C\sqrt{2}} \right)^{2/3}$, the expected normalized regret of \autoref{alg:ula-private-membership} is at most:
    \begin{equation}
    k\bar\rho_K + \frac{1}{\sqrt{\psi}} \left[ \left(\frac{3}{\pi} \right)^{2/3} \log P + \left( \frac{C\sqrt{2}}{N} \right)^{2/3} \right]^{3/2} = k\bar\rho_K + O\left( \frac{C}{N\sqrt{\psi}} \right)
\end{equation}
With $C$ defined as in Corollary \ref{cor:ula-sample-budget}, which matches the rate achieved by \autoref{alg:ula-direct} in the public unit membership case up to constant factors.
\end{lemma}

Because the irreducible error term from Definition \ref{def:excess-risk} is independent and zero-mean, ULA is able to effectively ignore it when taking averages across many individuals in a single unit, which is reflected in the following lemma:

\begin{lemma}\label{lem:ula-sco-baseline}
    Let $f$ be a binary classifier with expected $\ell_2^2$ loss $\alpha$, and let $\widehat \ULA$ denote the algorithm that allocates aid to units in ascending order based on $\hat\rho_j = \frac{1}{N_j}\sum_{i \in U_j} f(\hat x_i)$. Then in expectation, we have:
\begin{align}
   \overline\regret(\widehat \ULA, w; k) &\leq k\bar\rho_K + \sqrt{P |\mathscr{P}|\sigma^2} + P\sqrt{\alpha - \sigma^2}
\end{align}
\end{lemma}

\begin{theorem}\label{thm:ula-sco-private}
    Let $\widetilde\ULA$ be the algorithm that samples $n$ individuals from the population uniformly at random, fits a model $\hat\theta$ with expected $\ell_2^2$ error $\alpha$ subject to $\psi$-zCDP, and then runs \autoref{alg:ula-private-membership} on the estimated welfare scores of the remaining population with privacy parameter $\psi$ and partition $\mathscr{P}$. Then $\widetilde\ULA$ satisfies Joint $\psi$-zCDP in the replace-one adjacency definition with private unit membership, and:
    \begin{equation}
        \overline\regret(\widetilde\ULA, w; k) \leq k\bar\rho_K + P\sqrt{\alpha - \sigma^2} + \sqrt{P|\mathscr{P}|\sigma^2} + O\left(\frac{C}{N\sqrt{\psi}}\right)
    \end{equation}
\end{theorem}

Incorporating the cost of sampling and the utility guarantees from Theorem \ref{thm:feldman} yields the following corollary:

\begin{corollary}\label{cor:ula-sco-budget}
    When $\widetilde\ULA$ is instantiated with Algorithm 2 of \cite{feldman2020private}, its asymptotic normalized regret grows like:
    \begin{equation}
        k\bar\rho_K + (1-\bar\rho_K)n\lambda + \sqrt{P|\mathscr{P}|\sigma^2} + P\sqrt{E^* + 10LD \left( \frac{1}{\sqrt{n}} + \frac{\sqrt{p}}{n\sqrt{2\psi}} \right)} + O\left(\frac{C}{N\sqrt{\psi}}\right)
    \end{equation}
Choosing $n = \max\left( \left(\frac{P}{\lambda}\sqrt{10LD(2\psi)^{-1/2}\sqrt{p}}\right)^{2/3}, \left(\frac{P}{\lambda}\sqrt{10LD}\right)^{4/5} \right)$, we can guarantee a worst-case regret bound of:
\begin{align}
    &\quad k\bar\rho_K + P\sqrt{E^*} 
    + \sqrt{P|\mathscr{P}|\sigma^2} + O\left(\frac{C}{N\sqrt{\psi}}\right) \\
    &+ (2-\bar\rho_K)\left[ \left(\frac{P\sqrt{10LD\lambda\sqrt{p}}}{(2\psi)^{1/4}} \right)^{2/3}   
    + \lambda^{1/5} \left( P\sqrt{10LD} \right)^{4/5} \right]
\end{align}
    If we additionally have access to a lower bound on the \underline{excess} risk of the optimal classifier $E^* \geq E^\downarrow > 0$, then we can instead choose $n = \max\left( \sqrt{\frac{5PLD\sqrt{p}}{\lambda \sqrt{2\psi E^\downarrow}}}, \left(\frac{5PLD}{\lambda\sqrt{ E^\downarrow}}\right)^{2/3} \right)$ for an instance-specific regret bound of:
\begin{equation}
    k\bar\rho_K + P\sqrt{E^*} + \sqrt{P|\mathscr{P}|\sigma^2} + (2-\bar\rho_K) \left[ \sqrt{\frac{5PLD\lambda\sqrt{p}}{\sqrt{2\psi E^\downarrow}}} + \left( \frac{5PLD\lambda}{\sqrt{E^\downarrow}} \right)^{2/3} \right] + O\left(\frac{C}{N\sqrt{\psi}}\right)
\end{equation}

\end{corollary}

\subsection{Asymptotic Regimes}\label{sec:learning-regimes}

As in \Cref{sec:sampling}, the price of privacy is asymptotically negligible for both ILA and ULA, and so it suffices to compare the leading linear terms. For ULA, our regret grows like $k\bar\rho_K + P\sqrt{E^*} \leq k(1-G)\bar\rho_M + P\sqrt{E^*}$, while for ILA it grows like $P\sqrt{E^* + \sigma^2} + k'\bar\eta_{k'} \geq P\sqrt{E^* + \sigma^2}$. Using the first-order approximation $\sqrt{E^*} = \sqrt{E^* + \sigma^2 - \sigma^2} \leq \sqrt{E^* + \sigma^2} - \frac{\sigma^2}{2\sqrt{E^*+\sigma^2}}$, we obtain the following sufficient condition for ULA to outperform ILA asymptotically:
\begin{align}
    k(1-G)\bar\rho_M 
    &\leq \frac{P\sigma^2}{2\sqrt{E^*+\sigma^2}} \\
    2\frac{k}{P}(1-G)\bar\rho_M \sqrt{\mathcal{L}(\theta^*)} 
    &\leq \sigma^2
\end{align}

If the special case where $\Theta$ contains the Bayes' optimal classifier, $\mathcal{L}(\theta^*) = \sigma^2$ and the condition simplifies to:
\begin{equation}
    \left(\frac{1}{\sigma}\right) \left(\frac{k}{P} \right) (1-G)\bar\rho_M \leq \frac{1}{2}
\end{equation}
Each term on the left-hand side has a nice interpretation: $1/\sigma$ represents the intrinsic predictability of the outcomes we care about, $k/P$ represents the size of our budget, $1-G$ represents how \textit{equally} welfare is distributed across units, and $\bar\rho_M$ represents the average welfare across units. Taken as a whole, we see that ULA dominates ILA asymptotically whenever the product of these four terms is not too large.

\subsection{The Value of Modeling}

Since training a model to predict welfare implicitly requires the ability to sample welfare scores from the population, it would be possible to dispense with the modeling step entirely and use the algorithms from \Cref{sec:sampling} even in this stochastic setting. We therefore consider one final comparison between the regret bounds obtained for each strategy in this section and those for the same strategy in \Cref{sec:sampling}.

In the case of ILA, our linear regret term in the pure sampling setting scaled like $\frac{P\lambda k}{P\lambda + k}$, while the regret of the model-based approach scaled like $P\sqrt{\mathcal{L}(\theta^*)}$. Modeling therefore emerges as a stronger strategy in settings where $\lambda$ is relatively large. In particular, we saw that sampling-based strategies could do no better than random allocation for $\lambda > 1-\frac{k}{P}$, while modeling-based strategies can in principle achieve near-optimal regret even for $\lambda>1$.

In the case of ULA, we ignore the common linear term of $k\bar\rho_K$. Then, the excess regret of the sampling based approach grows like $O(P^2 M \lambda)^{1/3}$, while the excess regret of the modeling approach grows like $\lambda^{1/5}P^{4/5} + P\sqrt{E^*}$. We can therefore identify three natural regimes of interest:
\begin{itemize}
    \item If $E^* > 0$ and $M = o(P)$, then the modeling approach can \textit{never} outperform pure sampling asymptotically. The key reason for this is that the error from sampling is purely driven by variance, whereas the error from modeling can involve bias as well.
    \item If $E^* = 0$, then modeling can outperform sampling asymptotically only when $M = \Omega( (P/\lambda)^{2/5})$. This is because when $M$ is very large, the modeling based approach can leverage data from similar individuals across units to avoid sampling every unit directly.
    \item If $E^* > 0$ and $M = \Theta(P)$, then it is possible for modeling to outperform sampling if $\lambda \gtrsim (E^*)^{3/2}$, which once again occurs because modeling allows us to avoid sampling many individuals from every unit at the cost of introducing some bias.
\end{itemize}

It is interesting to note that the conditions for modeling to improve the performance of ULA seem to be considerably stricter than for ILA. One possible interpretation for this fact is that unit-level targeting in itself already constitutes a simple sort of model, and so the ability to leverage information about one individual to infer the welfare of another was already available to ULA even before making any distributional assumptions. In other words, ULA has less to gain than ILA from additional tools for making such inferences.

\section{Discussion}\label{sec:discussion}

We conclude our analysis by reflecting on the strengths and limitations of our modeling choices and suggesting possible avenues for future work.

\subsection{On Our Modeling Choices}
\label{sec:modeling-choices}

\subsubsection{Our Choice of Social Welfare Ordering}

For both ULA and ILA, the objective function we seek to maximize is $\sum_{i=1}^P w_i$, which corresponds to the classical Utilitarian social welfare ordering~\cite{moulin2004fair}. One limitation of this ordering is that it is indifferent to \textit{distributive} concerns; a world where half of the population has welfare 1 and the other half has welfare 0 is deemed just as desirable as a world where everybody has welfare 0.5. Possible alternatives that are more sensitive to distributive questions include the Nash social welfare ordering, which corresponds to maximizing $\prod_{i=1}^P w_i$, and the egalitarian social welfare ordering, which (roughly) seeks to maximize $\min_i w_i$.

The main challenge with both of these alternative objectives in the context of DP is that they can be extremely sensitive to changes in a single individual's welfare; one person with welfare 0 forces both objective functions to be 0 regardless of the remaining data. In contrast, the Utilitarian objective is a sum of bounded data points, which is a well-studied object in the DP literature. We therefore believe that the Utilitarian objective is the natural starting point for our investigation, and leave the design of DP algorithms for optimizing either of the other objectives (or choosing between multiple equal-value allocations~\cite{jain_allocation_2025}) to future work.

\subsubsection{Statistical Estimation vs. Pure Selection}

The strategies that we consider all derive their allocations from privatized estimates of important statistics, such as the empirical CDF of welfare scores or the unit profile vector. Outputting a private aid allocation is fundamentally a \textit{selection} task, however, and so one could argue that it would be more efficient to skip the estimation step and use methods, like those contained in in~\cite{gillenwater2022joint,bafna2017price,steinke2017tight,jourdan2025optimal}, which optimize over the set of possible allocations directly. In the specific domain of aid allocation, however, we believe that approaches based on statistical estimation have major advantages in transparency: in addition to outputting an allocation, our algorithms can release a privatized statistic which justifies why that allocation is fair. A further benefit is that the intermediate statistics we compute are general enough to be repurposed for other tasks without expending any additional privacy budget.

\subsubsection{What Our Model Does Not Capture}

Our model assumes a very simple and schematic relationship between welfare and treatment effects which may not hold exactly in real settings. The authors of \cite{shirali_allocation_2024} on whose work we build upon consider more complicated treatment effects and show that when $\tau(w)$ is Lipschitz and concave, and when the distribution of welfare scores is sufficiently smooth, estimating welfare alone is still sufficient to derive efficient allocations. An alternative approach, which others have advocated for on philosophical grounds, is to reframe the problem in terms of predicting treatment effects directly~\cite{barabas2018interventions,liu2023reimagining}. To this end, there is a large body of works, spanning many sciences, which study how heterogeneous treatment effects can be learned~\cite{wager2018estimation,kunzel2019metalearners}, including under DP constraints~\cite{niu2022differentially}. Although we do not directly pursue this point, we remark that one virtue of the simplicity of our model is that our results can be easily restated to take either treatment effects or welfares as primary.

\subsection{Conclusion and Future Work}

In this work, we separately considered two possible constraints on data availability and found that in all cases, both ULA and ILA could be modified to satisfy DP without affecting their asymptotic regret. In \Cref{sec:sampling}, we imagined that welfare scores were fixed arbitrarily and designed robust algorithms for achieving low regret in the worst-case. Here, the key difference between ILA and ULA was that ULA required only a sublinear number of samples to achieve nearly-optimal targeting, while ILA was obliged to expend a constant fraction of its budget on sampling. As a result, the optimal choice of strategy depended strongly on $\lambda$, the relative cost of measuring welfare compared to the cost of treatment. 

Meanwhile, in \Cref{sec:learning}, we assumed that welfare scores and features were sampled i.i.d. from some fixed joint distribution and designed learning-based algorithms that could estimate unseen welfare scores with high accuracy. Unlike the previous setting, we found that there was essentially no difference between the sampling required by ILA vs. ULA; instead, the choice of optimal strategy depended on a new quantity, $\sigma^2$, representing the fundamental unpredictability of welfare scores conditioned on the available data. Because ULA was able to ``wash out'' this uncertainty by averaging many measurements together, it was significantly more robust than ILA in settings where outcomes are fundamentally difficult to predict.

A natural next step (with or without DP) is to combine these two settings as follows: imagine an administrator trying to select an allocation strategy. To start out, they have access to relatively coarse features that encode unit membership and nothing else. But by paying some cost $c_p$, they can obtain $p$ new features for a single individual --- for instance, performance on standardized assessments~\cite{ikram_frugal_nodate} or proxies for socioeconomic status~\cite{kohli_enabling_2024}. By expanding their model class in this way, the administrator can potentially decrease the irreducible error of their prediction problem, but this will come at the cost of increasing their estimation error and obliging them to spend more of their budget on acquiring extended feature vectors. Seen in this light, ILA and ULA might appear less as diametric opposites, and more as two points along a continuum of tradeoffs between the error arising from coarse targeting, sampling, and epistemic uncertainty.

An alternative path to extend our model is to incorporate adaptive or strategic behavior. On the administrative side, this might involve conducting repeated rounds of sampling and/or intervention, which could enable meaningful reductions in sample complexity using techniques from the bandits literature~\cite{audibert2010best,jourdan2025optimal} or permit treatment effects and allocations to be learned jointly~\cite{kasy2021adaptive,wilder2025learning}. On the population side, this might involve considering the possibility of strategic behavior from both individuals~\cite{bruckner2011stackelberg} and organized groups~\cite{hardt_algorithmic_nodate} wishing to influence the administrator's decision. It would be particularly interesting to extend our model to incorporate insights from the large body of literature on \textit{performative prediction}~\cite{sperber_looping_1995,manheim_categorizing_2019,hardt_performative_2023,hardt_performative_2022}, whereby forecasts can influence the very events they predict, and to investigate how the stability induced by DP impacts these dynamics.


\newpage

\bibliography{references}
\bibliographystyle{plainurl}


\appendix

\section*{Appendices}

\section{Differences between Our Model and that of Shirali et al.~\cite{shirali_allocation_2024}}

In their definition of ULA, the authors of \cite{shirali_allocation_2024} make the additional assumption that unit-level administrators can be relied on to allocate aid \textit{within} each unit efficiently, e.g. by avoiding the top $q$ fraction of individuals by welfare with probability $1-q'$. 

We regard this assumption as unsatisfying for two main reasons. Firstly, the allocation scheme described essentially corresponds to a hybrid of ILA and ULA, which could be called ``ILA within units.'' However, the treatment of ILA within units and full ILA is asymmetric --- their model assumes that ILA is forced to expend some of its budget to guarantee accurate targeting across the population, but there is no explanation or cost associated with the accuracy of within-unit targeting. We find this asymmetry unsatisfying. Secondly, in our specific context of \textit{private} allocations, we require the entire process by which allocation decisions are made to satisfy DP. It is therefore important that any decision-making power exercised by unit-level administrators be modeled explicitly to allow for rigorous privacy analysis. To resolve these issues, our model assumes that ULA allocates aid uniformly to all individuals within each unit, without targeting (see Section \ref{sec:ula}).

We recognize that the use of within‑unit targeting in~\cite{shirali_allocation_2024} was partly motivated by the idea that allocating aid to every member of a unit is clearly inefficient in certain contexts. For instance, when units correspond to schools, it may be reasonable to assume that administrators won't allocate limited tutoring resources to honor‑roll students; in this context, within-unit targeting typically wouldn't raise additional privacy concerns because the identities of honor-roll students are often publicly available to begin with. 

We believe that the cleanest way to reconcile this tension and make our model of aid allocation compatible with DP is to assume that targeting within units is uniform while defining units at a finer level of granularity (e.g., ``honor‑roll students at School A,'' ``non‑honor‑roll students at School B,'' etc.), a strategy we examine in greater detail in \Cref{sec:learning}. More generally, we assume that the population under consideration for both ILA and ULA doesn't include individuals who are \textit{a priori} ineligible for a benefit, so that resources are not expended on estimating the welfare of irrelevant units.

\section{Useful Technical Results}\label{app:useful}

The following two lemmas provide high‑probability upper bounds on the maximum of Gaussian and hypergeometric random variables, which we utilize in our later proofs in the Appendix.

\begin{lemma}[Concentration of maximum of Gaussians]\label{lem:gaussian-concentration}
    Let $z_1,\ldots, z_n$ be (not-necessarily independent) Gaussian random variables with $z_i \sim \mathcal{N}(0, \sigma_i^2)$, and let $\sigma_*^2 = \max_i \sigma_i^2$. Then, with probability at least $1-\beta$, we have:
    \begin{equation}
        \max_i z_i \leq \sigma_* \left(\sqrt{2\log n} + \sqrt{2\log(1/\beta)}\right)
    \end{equation}
\end{lemma}

\begin{proof}
    Let $Z = \max_i z_i$. Using the standard upper-bound based on moment-generating functions, we have that $\mathbb{E}[Z] \leq \sigma_* \sqrt{2\log n}$.
    Then, by the Borell-TIS inequality, we know that the maximum of Gaussians concentrates closely around its expectation~\cite{adler2007gaussian}:
    \begin{equation}
        \mathbb{P}(
        Z > \mathbb{E}[Z] + u
        ) \leq \exp\left( \frac{-u^2}{2\sigma_*^2} \right)
    \end{equation}
   Combining these two results and rearranging yields the lemma.
\end{proof}

\begin{lemma}[Concentration of hypergeometric random variables (\cite{serfling1974probability} Corollary 1.1)]\label{lem:hypergeom-concentration}
    Let $X$ be a Hypergeometric random variable with population size $P$, sample size $k$, and mean $kp$. Then with probability at least $1-\beta$,
    \begin{equation}
        |X/k - p| \leq \sqrt{\frac{P-k}{P}} \sqrt{\frac{\log(2/\beta)}{k}}
    \end{equation}
\end{lemma}

A fundamental tool for achieving zCDP is the \textit{Gaussian mechanism}, which we use as a building block in our algorithms.

\begin{lemma}[$\psi$-zCDP Guarantee of the Gaussian Mechanism (\cite{bun2016concentrated} Proposition 16)]\label{lem:gaussian-mechanism}
  Let $q: \mathcal{X}^n \to \mathbb{R}$ be a sensitivity $\Delta$ query. Then the mechanism $M: \mathcal{X}^n \to \mathbb{R}$ defined by $M(x) \sim \mathcal{N}(q(x),~\Delta^2/(2\psi))$ satisfies $\psi$-zCDP.
\end{lemma}

In \Cref{sec:learning}, we use the following result of Feldman et al.~\cite{feldman2020private} to characterize the sample complexity of DP-SCO:

\begin{theorem}[Sample Complexity of DP-SCO (\cite{feldman2020private} Theorem 4.4)]\label{thm:feldman} Let $\mathcal{X} \times [0,1]$ be the domain of datasets, and $\mathcal{D}$ be a distribution over $\mathcal{X} \times \lbrace 0,1 \rbrace$. Let $S = ((x_1, y_1) \ldots, (x_n, y_n))$ be a dataset drawn i.i.d. from $\mathcal{D}$. Let $\mathcal{K} \subseteq \mathbb{R}^p$ be a convex set denoting the space of all models. Let $\ell: \mathcal{K} \times \mathcal{X} \times \lbrace 0,1 \rbrace \to \mathbb{R}$ be a loss function, which is $L$-Lipschitz, $\xi$-smooth, and convex in its first parameter. Then, if $\mathcal{K}$ has diameter at most $D$, the output $\hat \theta$ of Algorithm 2 in \cite{feldman2020private} satisfies $\psi$-zCDP and we have:
\begin{equation}
    \mathbb{E}[\mathcal{L}(\hat \theta)] \leq \min_{\theta^* \in \mathcal{K}} \mathcal{L}(\theta^*) + 10LD \left( \frac{1}{\sqrt{n}} + \frac{\sqrt{p}}{n\sqrt{2\psi}} \right)
\end{equation}
where $\mathcal{L}(\theta) = \mathbb{E}_{x,y \sim \mathcal{P}}[\ell(\theta,x,y)]$, provided that $\frac{D}{L}\min \left\lbrace \frac{4}{\sqrt{n}}, \frac{\sqrt{2\psi}}{\sqrt{p}} \right\rbrace \leq 2/\xi$.
\end{theorem}

\section{Missing Proofs from \Cref{sec:prelim}}

\subsection{Proof of Lemma \ref{lem:ila-helper}}

Let $k' \leq k$ be given. As in the lemma statement, let $\mathcal{I}$ denote the set of individuals mistakenly included or excluded from our $k'$-allocation because of estimation error. Maximizing the normalized expected value of ILA can be represented as a linear program:
\begin{align}
    \max_x  \sum_{i=1}^P x_i (1-\eta_i) \quad s.t. \quad &\sum_{i=1}^P x_j = k' \\
    &0 \leq x_i \leq 1 \quad \forall i
\end{align}

Since the optimal solution for an integer $k'$ is to assign $x_{s(i)} = 1$ for all $i \leq k'$, we do not need to worry about non-integral solutions.

Let $x^*$ denote the optimal solution for the true $\eta$ vector, and $\tilde x^*$ be the optimal solution for the perturbed $\widehat\eta$ induced by our estimated probabilities. Then, letting $\langle \cdot, \cdot \rangle$ denote the standard inner product, we have:
\begin{align}
   \langle \tilde x^*, 1-\eta \rangle 
   &= \langle \tilde x^*, 1-\widehat\eta \rangle - \langle \tilde x^*, \eta - \widehat\eta \rangle \\
   &\geq \langle x^*, 1-\widehat\eta \rangle -  \langle \tilde x^*, \eta - \widehat\eta \rangle\\
   &= \langle x^*, 1-\eta\rangle - \langle x^*, \widehat\eta-\eta \rangle - \langle \tilde x^*, \eta - \widehat\eta \rangle \\
   &= \langle x^*, 1-\eta\rangle - \langle x^* - \tilde x^*, \widehat\eta - \eta \rangle \\
   &= \langle x^*, 1-\eta \rangle - \sum_{i=1}^P (x_i^* - \tilde x_i^*) (\widehat\eta_i - \eta_i) 
\end{align}

Subtracting $\langle \tilde x^*, 1-\eta \rangle$ and adding $\sum_{i=1}^P (x_i^* - \tilde x_i^*) (\widehat\eta_i - \eta_i)$ to both sides gives us that our regret is at most:
\begin{equation}
    \sum_{i=1}^P (x_i^* - \tilde x_i^*) (\eta_i - \widehat\eta_i) \leq \sum_{i \in \mathcal{I}} |\eta_i - \widehat\eta_i|
\end{equation}

From here, we have that in expectation, the regret of $x^*$ with respect to the optimal $k'$-allocation in hindsight (i.e. based on the realized welfares) is exactly $\sum_{i=1}^P x_i \eta_i = \sum_{i=1}^{k'} \eta_{s(i)} \eqqcolon \bar \eta _{k'}$.

Finally, the normalized regret with respect to the optimal $k$-allocation is then at most that quantity plus $(k-k')$.

\subsection{Proof of Lemma \ref{lem:ula-helper}}

In Shirali et al. \cite{shirali_allocation_2024} (Lemma 4.1), the authors prove that the value of ULA is lower-bounded by $\sum_{j=1}^{k/N} N_{s(j)}(1-\rho_{s(j)})$, implying that $\overline\regret(\ULA, w; k) \leq \sum_{j=1}^{k/N} N_{s(j)} \rho_{s(j)}$. An immediate generalization says that if we allocate aid to $x_j \in [0, N_j]$ individuals chosen uniformly at random in unit $j$, then our normalized regret is upper-bounded by $\sum_{j=1}^M x_j \rho_j$. This allows us to formulate ULA as a linear program:
\begin{align}
    \min_x \sum_{j=1}^M x_j \rho_j \quad s.t. \quad &\sum_{j=1}^M x_j = k \\
    &0 \leq x_j \leq N_j\quad \forall j
\end{align}

Let $x^*$ denote the optimal solution for the true unit profile vector $\rho$, and let $\tilde x^*$ denote the optimal solution for $\widetilde\rho$. Then:
\begin{align}
   \langle \tilde x^*, \rho \rangle 
   &= \langle \tilde x^*, \widetilde\rho \rangle + \langle \tilde x^*, \rho - \widetilde\rho \rangle \\
   &\leq \langle x^*, \widetilde\rho \rangle +  \langle \tilde x^*, \rho - \widetilde\rho \rangle\\
   &= \langle x^*, \rho\rangle + \langle x^*, \widetilde\rho-\rho \rangle + \langle \tilde x^*, \rho - \widetilde\rho \rangle \\
   &= \langle x^*, \rho\rangle + \langle x^* - \tilde x^*, \widetilde\rho - \rho \rangle \\
   &= \langle x^*, \rho \rangle + \sum_{j=1}^M (x_j^* - \tilde x_j^*) (\widetilde\rho_j - \rho_j) \\
   &\leq \langle x^*, \rho \rangle + \lVert x^* - \tilde x^*\rVert_p \lVert \widetilde\rho - \rho\rVert_q
\end{align}
where the first inequality follows from the optimality of $\langle \tilde x^*, \widetilde\rho \rangle$ and the last line holds for any $p, q \in [1, \infty]$ such that $1/p + 1/q = 1$ by H\"older's inequality. The last step to derive the general bound is to note that on a worst-case input where all individuals have welfare $w \in \lbrace 0, 1\rbrace$, the inequality $\langle x^*, \rho \rangle \geq \overline\regret(\ULA, w, k)$ is in fact tight.

For the specific choices of $p, q \in \lbrace1, 2, \infty \rbrace$, we argue as follows. Let $\vec N = (N, N, \ldots, N) \in \mathbb{R}^M$. By the triangle inequality, for any $p \in [1, \infty]$, we have the two following upper bounds: 
\begin{align}
     &\lVert x^* - \tilde x^* \rVert_p \leq \lVert x^* \rVert_p + \lVert \tilde x^* \rVert_p  \\
    &\lVert x^* - \tilde x^* \rVert_p = \lVert (x^* - \vec N) - (\tilde x^* - \vec N) \rVert_p \leq \lVert x^* - \vec N \rVert_p + \lVert \tilde x^* - \vec N \rVert_p
\end{align}

And so in particular, we can always choose the smaller of the two upper bounds. For $p=1$, we have $\lVert x^* \rVert_1 = k$ and $\lVert x^* - \vec N \rVert_1 = P-k$, giving a final symmetric upper bound of $2\min(k, P-k)$. For $p=2$, we have $\lVert x^* \rVert_2 \leq \sqrt{kN}$ because $x^*$ has at most $k/N$ components with squared magnitude $N^2$, and symmetrically $\lVert x^* - \vec N \rVert_2 \leq \sqrt{(P-k)N}$. Combining these results yields the desired bounds:
\begin{equation}
        \begin{cases}
           2\min(k,P-k) \max_j |\widetilde\rho_j - \rho_j| & p=1, q=\infty \\
           \sqrt{2\min(k,P-k)N \sum_{j=1}^M (\widetilde\rho_j - \rho_j)^2}   & p=2, q=2 \\
           N \sum_{j=1}^M |\widetilde\rho_j - \rho_j|  &p = \infty, q=1
        \end{cases}
    \end{equation}

\subsection{Proof of Lemma \ref{lem:random-helper}}

Let $X$ denote the number of low-welfare individuals selected by \rand. Then $\val(\rand) \geq \delta_w X$, which implies that $\max_{|\mathcal{J}|=k} \val(\mathcal{J}, w) - \val(\rand, w) \leq  \max_{|\mathcal{J}|=k} \val(\mathcal{J}, w) - \delta_w X \leq \delta_w k - \delta_w X$, and therefore that $\overline\regret(\rand, w; k) \leq k-X$. Observing that $X$ follows a hypergeometric distribution completes the proof.

\subsection{Proof of Lemma \ref{lem:gini-baseline}}

The proof proceeds by maximizing $G$ subject to equality constraints on $\bar \rho_M$ and $\bar\rho_K$. For simplicity, we ignore the constraint that $\rho_i \in [0,1]$; this can only increase the value of the objective function and so our final bound remains valid. Without this constraint, it is clear by inspection of the second representation of the Gini coefficient, $G \propto \sum_{1 \leq i \leq j \leq M} (\rho_{s(j)} - \rho_{s(i)})$, that it is maximized when $M-1$ units have the same welfare and a single unit has higher welfare. So, we set:
\begin{align}
    \rho_{s(i)} &= a\qquad\forall i \leq M-1 \\
    \rho_{s(M)} &= a + b \\
    \bar\rho_M &= a + \frac{b}{M} \\
    \bar\rho_K &= a \\
    \Delta_K &\coloneq \bar\rho_M - \bar\rho_K = \frac{b}{M}
\end{align}

Then, using the definition of the Gini coefficient, we have:
\begin{align}
    G = \frac{1}{M^2 \bar\rho_M} (M-1)b = \frac{1}{M^2\bar\rho_M}(M-1)\frac{Mb}{M} = \frac{1}{M\bar\rho_M}(M-1)\Delta_K \\
    \Rightarrow \Delta_K = \frac{GM\bar\rho_M}{M-1} > G\bar\rho_M
\end{align}

as desired.

\section{Missing Proofs from \Cref{sec:direct}}

\subsection{\texttt{PrivatePartialSums}}
\label{app:private-partial_sums}

There is an extensive body of work on the problem of privately computing partial sums of bounded data points, which is described in detail in ~\cite{pillutla2025correlated}. We use \texttt{PrivatePartialSums} to refer to the factorization-based algorithm studied in ~\cite{henzinger_almost_2023}, whose utility guarantees are summarized in~\cite[Theorem 2.7]{pillutla2025correlated}. Briefly, their algorithm is based on adding correlated Gaussian noise to each of the $n$ partial sums of the input sequence with maximum standard deviation:
\begin{equation}
    \frac{1}{\sqrt{\psi}} \left(1 + \frac{\log(n) + \gamma_E}{\pi} \right)
\end{equation}
where $\gamma_E\approx 0.577$ denotes Euler's constant, which, when combined with Lemma \ref{lem:gaussian-concentration}, yields the error term in Theorem \ref{thm:ila-direct-helper}.

\subsection{Proof of Lemma \ref{lem:ila-privacy}}

In \autoref{alg:ila-direct}, individual data first appears in Step 5 when executing \texttt{PrivatePartialSums}, which satisfies $\psi$-zCDP. Step 6 operates solely on the output of Step 5, so by post-processing this step also satisfies $\psi$-zCDP \cite{bun2016concentrated}. Lastly, by the Billboard lemma (Lemma \ref{lem:billboard}), the return statement of Step 7 satisfies Joint $\psi$-zCDP.

\subsection{Proof of Lemma \ref{lem:ila-direct-budget}}

This lemma follows from the utility guarantee of \texttt{PrivatePartialSums} (described above) along with Lemma \ref{lem:gaussian-concentration}, which together tell us that with probability at least $1-\beta/2$,
\begin{equation}
    \max_j |S_j - \tilde S_j| \leq \frac{1}{\sqrt{\psi}} \left(1 + \frac{\log(n) + \gamma_E}{\pi} \right) \left(\sqrt{\log\left(\frac{1+2s}{\theta}\right)} + \sqrt{\log\frac{2}{\beta}}\right)
\end{equation}
Conditioning on this event, $j^*$ is the index of the first bin where it is possible that $S_{j^*} \geq k$. In particular, this means that $S_{j^*-1} < k$. Because we define our intervals to have exclusive left endpoints, it follows that the final threshold $\ell_{j^*}$ will not cause us to go over budget.

\subsection{Proof of Theorem \ref{thm:ila-direct-helper}}

We begin by discretizing the interval $\mathscr{R} = [-s, 1+s]$ into $(1+2s)/\theta$ bins, $\lbrace (\ell_i, r_i] \rbrace_{i=1}^{(1+2s)/\theta}$ of uniform width $\theta$ (as a special case, the very first interval should be inclusive on both sides, i.e. $[-s,-s+ 1/\theta]$). Let $Y_{i,j} = \mathbb{I}[\hat w_j \in B_i]$, and define the partial sum $S_i = \sum_{q=1}^i\sum_{j=1}^n Y_{q,j}$. Using \texttt{PrivatePartialSums} (see Appendix \ref{app:private-partial_sums}), we can privately compute estimates of all of the $S_i$'s simultaneously with $\tilde{S}_i \sim \mathcal{N}(S_i, \sigma_i^2)$ and $\sigma_i^2 \leq \sigma_{(1+2s)/\theta}^2 \approx \frac{1}{\psi} \left(1 + \frac{\log((1+2s)/\theta) + \gamma_E}{\pi}\right)^2$ where $\gamma_E\approx 0.577$ denotes Euler's constant.

Using Lemma \ref{lem:gaussian-concentration}, we have that with probability at least $1-\beta/2$, 
$$
\max_i |\tilde{S}_i - S_i| \leq \sigma_{(1+2s)/\theta}(\sqrt{2\log((1+2s)/\theta)} + \sqrt{2\log(2/\beta)}) \eqqcolon t
$$ 
From here, we select $j^* = \min \lbrace j \mid \tilde{S}_j + t \geq k \rbrace$ and choose our threshold as $\ell_{j^*}$.

We analyze the chosen threshold. Firstly, conditioning on our high-probability bound for the error from the noise we add, we know that the interval $[-s, r_{j^*}]$ must contain at least $\tilde{S}_{j^*} - t$ points in our dataset. From the definition of $j^*$, it follows that it therefore contains least $k - 2t$ points. Therefore, the total number of individuals that we exclude is at most $2t$ plus the number of individuals who fell in the interval $(\ell_{j^*}, r_{j^*}]$.



By our assumptions, the probability that a given datapoint falls in a certain interval is at most $p$. So, we can stochastically upper-bound the count of each bin as a Gaussian random variable with mean at most $np$ and variance at most $np(1-p)$.
Once again using Lemma \ref{lem:gaussian-concentration}, we have that with probability at least $1-\beta/2$:
\begin{equation}
    \max_{1 \leq i \leq (1+2s)/\theta} S_i - S_{i-1} \leq np + \sqrt{2np}\left(\sqrt{\log((1+2s)/\theta)} + \sqrt{\log(2/\beta)} \right)
\end{equation}

With the convention that $S_0 = 0$. Finally, since treatment effects are $1$-Lipschitz in welfare by definition, the additional regret from choosing $s>0$ is at most $ks$. Combining the error terms from privately computing partial sums, choosing a conservative threshold, and additive noise with $s>0$ yields the theorem statement.

\subsection{Proof of Theorem \ref{thm:ula-direct}}

We specialize Lemma \ref{lem:ula-helper} to the case where each $\widetilde\rho_j - \rho_j$ is an independent Gaussian with variance $\sigma_j^2$. Define $\sigma^2 = \frac{1}{M} \sum_{j=1}^M \sigma_j^2$, $\bar\sigma = \frac{1}{M} \sum_{j=1}^M \sigma_j$, and $\sigma_*^2 = \max_j \sigma_j^2$. Then by Lemma \ref{lem:gaussian-concentration}, we have:
\begin{equation}
    \mathbb{E}[\max_j |\widetilde\rho_j - \rho_j]] = \mathbb{E}[\max_{j,\epsilon = \pm 1} \epsilon(\widetilde\rho_j - \rho_j)] \leq \sigma_* \sqrt{2\log(2M)}
\end{equation}
And by Jensen's inequality and additivity of variance, we have that:
\begin{equation}
    \mathbb{E}\left[ \sqrt{\sum_{j=1}^M (\widetilde\rho_j - \rho_j)^2} \right] \leq \sqrt{\mathbb{E}\left[ \sum_{j=1}^M (\widetilde\rho_j - \rho_j)^2\right]} = \sigma\sqrt{M}
\end{equation}
Finally, standard properties of truncated normals give us that:
\begin{equation}
    \mathbb{E}[|\widetilde\rho_j - \rho_j|] = \sigma_j \sqrt{\frac{2}{\pi}}
\end{equation}
So, for this special case of Gaussian error, our regret bounds become:

\begin{equation}
    \begin{cases}
           \sigma_* 2\min(k,P-k)  \sqrt{2\log(2M)} & p=1, q=\infty \\
           \sigma \sqrt{2\min(k,P-k)P}   & p=2, q=2 \\
           \bar\sigma P \sqrt{\frac{2}{\pi}} & p=\infty, q=1
    \end{cases}
\end{equation}

In general, $\bar\sigma \leq \sigma \leq \sigma_*$, but if all units have equal variance then $\bar\sigma = \sigma = \sigma_*$. Comparing the first two bounds, we see that $p=1$ is tighter when:
\begin{equation}
    \min(k, P-k) \leq \frac{\sigma^2P}{4\sigma_*^2\log(2M)}
\end{equation}
i.e. it should be preferred for very small or very large $k$, especially when variances are equal or close to equal across units. Comparing the second and third bounds, we see that $p=\infty$ is preferable when:
\begin{equation}
    \min(k, P-k) \geq \frac{\bar\sigma^2P}{\sigma^2 \pi}
\end{equation}
i.e. it should be preferred when $k$ is close to $P/2$, especially when variances are very unequal across units. For the special case where all $\sigma_j$ are equal, we get the following simplified regret bound:
\begin{equation}
    \sigma \cdot \min\left(2\min(k,P-k)\sqrt{2\log(2M)}, \sqrt{2\min(k, P-k)P}, P\sqrt{2/\pi}\right)
\end{equation}

Finally, all of the norms we consider are Lipschitz functions of Gaussian random variables, which implies that they are concentrated exponentially around their means. For $p=2,\infty$, the Lipschitz constant is 1, and so with probability at least $1-\beta$, we have $\lVert \widetilde\rho - \rho \rVert_p \leq \mathbb{E}[\lVert \widetilde\rho - \rho \rVert_p] + \sigma\sqrt{2\log(1/\beta)}$. For $p=1$, the Lipschitz constant is instead $\sqrt{M}$, and so we have a looser bound of $\lVert \widetilde\rho - \rho \rVert_1 \leq \mathbb{E}[\lVert \widetilde\rho - \rho \rVert_1] + \sigma\sqrt{2M\log(1/\beta)}$.

\section{Missing Proofs from \Cref{sec:sampling}}

\subsection{Proof of Theorem \ref{thm:ila-sample-budget}}

For the remainder of this section, we fix $k = B/c$ and define the ratio $\lambda = c_2/c$. Our goal is to select parameters $n,k'$ to minimize our regret relative to the best $k$-allocation, subject to the budget constraint that $ck' + c_2n = B \Rightarrow k' + \lambda n = k$.  We model this problem as an interactive game as follows:

\begin{enumerate}
    \item \textbf{Fix input.} An adversary chooses a distribution $\mathcal{D}$ and samples $w \in [0,1]^{P} \sim \mathcal{D}$
    \item \textbf{Sample.} We choose an index set $|\mathcal{I}| = n$ arbitrarily and get to observe $w_i$ for each $i \in \mathcal{I}$
    \item \textbf{Allocate.} Using this information, we output a new index set $|\mathcal{J}| = k' \coloneq k - \lambda n$ and suffer regret $\textsc{Regret}(\mathcal{J}, w; k)$
\end{enumerate}

Our goal is to derive minimax-optimal regret bounds for this game, where the regret of algorithm $A$ on input $w$ with problem size $P,k,\lambda$ is defined by: 
\begin{equation}
    \textsc{Regret}(A, w) = \max \left \lbrace \sum_{j \in \mathcal{J}} \tau_j ~\bigg|~ |\mathcal{J}| = k \right\rbrace - \sum_{i \in A(w)} \tau_i 
\end{equation}

We define $\mathcal{A}$ to be the set of possible deterministic algorithms for the problem with fixed input size $(P,k,\lambda)$ and $\mathcal{R}$ to be the set of randomized algorithms over $\mathcal{A}$. Then, for any specific distribution over inputs $\mathcal{D}$ and any specific randomized algorithm $R'$, we have the following chain of inequalities:
\begin{align}
    \min_{A \in \mathcal{A}} \mathbb{E}_{w \sim \mathcal{D}}  \textsc{Regret}(A, w) 
    &\leq \max_{\mathcal{D}'} \min_{A \in \mathcal{A}} \mathbb{E}_{w \sim \mathcal{D}'} \textsc{Regret}(A, w) \\
    &= \min_{R \in \mathcal{R}} \max_{w} \mathbb{E}_{A \sim R} \textsc{Regret}(A, w) \\
    &\leq \max_{w} \mathbb{E}_{A \sim R'} \textsc{Regret}(A, w)
\end{align}

The two inequalities follow from the definition of $\min$ and $\max$, while the equality follows from Yao's minimax principle~\cite{yao1977probabilistic}. From here, we produce nearly matching lower and upper bounds by constructing specific $\mathcal{D}$ and $R'$.

\subsubsection{Upper Bound}

Let $R$ be a randomized algorithm with a single parameter $\kappa \in [k']$ that operates as follows: first, choose $n$ indices $\mathcal{I}$ uniformly at random. Then, select the $\kappa$ smallest revealed elements. Choose the remaining $k' - \kappa$ elements by randomly sampling the unseen indices. Our approach will be to upper-bound the regret of this algorithm by counting the number of elements it selects that are not part of the optimal index set $\mathcal{J}^*$, which is sufficient to bound regret because each element must be in the bounded set $[0,1]$.

For a fixed $w$, let $w_{s(k)}$ denote the value of the $k$th smallest element of $w$, so that $\mathcal{J}^* = \lbrace s(i) \mid i \in [k] \rbrace$. Then, the probability that a randomly sampled index is in $\mathcal{J}^*$ is $\frac{k}{P}$. So, in expectation, $\mathcal{I}$ will contain $\frac{nk}{P}$ such values. Let $Y$ denote the number of such values, which follows a hypergeometric distribution. Then our expected regret is equal to:

\begin{align}
    &\mathbb{E}\left[ k - \min(\kappa, Y) - \frac{(k'-\kappa)(k-Y)}{P-n} \right] \\
    &\leq k - \min\left(\kappa, \frac{nk}{P} \right) + \mathbb{E}\left[\max\left(0, -Y + \frac{nk}{P} \right)\right] - \mathbb{E}\left[ \frac{(k'-\kappa)(k-Y)}{P-n}\right]
    \\
    &\leq k - \min\left(\kappa, \frac{nk}{P} \right) + \sqrt{\frac{nk(P-k)(P-n)}{4P^3}} - \frac{(k'-\kappa)(k-\frac{nk}{P})}{P-n} \\
    &\leq \max\left(k-\kappa, k\left(1 - \frac{n}{P} \right)\right) - \frac{(k - n\lambda - \kappa)(k-\frac{nk}{P})}{P-n} + \sqrt{\frac{\min(n, P-n)}{16}} \\
    &= \max\left(k-\kappa, k\left(1 - \frac{n}{P} \right)\right) - \frac{(k - n\lambda - \kappa)k}{P} + \sqrt{\frac{\min(n, P-n)}{16}}
\end{align}

This is a sum of a piece-wise linear function of $n,\kappa$ and a concave function of $n$ with linear constraints, and so it must obtain its minimum at an extreme point or an elbow. That is, we must either have $n \in \lbrace 0, \min(P, k/\lambda)\rbrace$ and $\kappa \in \lbrace \max(0, k'+n-P), \min(k', n) \rbrace$, or else $\kappa = \frac{kn}{P}$.

If $n=P$, then $\kappa$ must be equal to $k'$, in which case we obtain a regret bound of $k-k'$.
If $n = k/\lambda$, then $\kappa = 0$ and our regret is $k$.
If $n=0$, then $\kappa$ must be equal to $0$ as well, and we obtain a regret bound of $k - \frac{k^2}{P}$.
Finally, if $\kappa = \frac{kn}{P}$, then our regret becomes:
\begin{align}
    k-\frac{kn}{P} - \frac{(k-n\lambda-\frac{kn}{P})k}{P} + \sqrt{\frac{\min(n, P-n)}{16}}
\end{align}
which is once again a linear function of $n$. Choosing $n=0$ from this point recovers our earlier bound; the maximum feasible point corresponds to $n = \frac{Pk}{P\lambda + k}, \kappa = k' = \frac{k^2}{P\lambda +k}$. So, there are only three plausible minimizers for our regret, which correspond to sampling the entire population, allocating at random, or something in the middle:
\begin{equation}
    \begin{cases}
        \lambda P & n = P \\
        \frac{P\lambda k}{P\lambda+k}+ \sqrt{\frac{\min(Pk, P^2\lambda)}{16(P\lambda+k)}} & n = \frac{Pk}{P\lambda + k}\\
        k - \frac{k^2}{P} & n=0
    \end{cases}
\end{equation}
Sampling the entire population is clearly never optimal because $\frac{k}{P\lambda + k} < 1$. Choosing $n = \frac{Pk}{P\lambda + k}$ becomes preferable asymptotically when:
\begin{align}
    \frac{P\lambda k}{k+P\lambda} 
    &\leq 
    k - \frac{k^2}{P} 
    \\
    P^2 \lambda k
    &\leq
    k^2P + kP^2\lambda - k^3 - k^2P\lambda
    \\
    \lambda 
    &\leq 
    \frac{k^2(P-k)}{k^2P}
    \\
    \lambda 
    &\leq 
    1 - \frac{k}{P}
\end{align}
giving us a final upper bound of:
\begin{align}
    \begin{cases}
        \frac{P\lambda k}{P\lambda+k} + \sqrt{\frac{\min(Pk, P^2\lambda)}{16(P\lambda+k)}} & \lambda < \frac{P-k}{P} \\
        k - \frac{k^2}{P} & \lambda \geq \frac{P-k}{P}
    \end{cases}
\end{align}

\paragraph{Concentration of Regret} Whichever regime we find ourselves in, our upper-bound on regret is equal to a constant plus a 1-Lipschitz function of a Hypergeometric random variable $Y$. When $\lambda \geq \frac{P-k}{P}$, $Y$ has population size $P$, sample size $k$, and number of success $k$, and so we can use Lemma \ref{lem:hypergeom-concentration} to conclude that with probability at least $1-\beta$,
\begin{align}
    \textsc{Regret}(A, w)
    \leq \mathbb{E}[\regret] + \sqrt{\frac{P-k}{P}}\sqrt{k\log(1/\beta)} = \sqrt{\frac{k(P-k)}{P}\log(1/\beta)}
\end{align}

Symmetrically, when $\lambda < \frac{P-k}{P}$, $Y$ has population size $P$, sample size $n=\frac{Pk}{P\lambda+k}$, and number of successes $k$, and so with probability at least $1-\beta$,
\begin{align} 
\textsc{Regret}(A, w)
    &\leq \mathbb{E}[\regret] + \sqrt{\frac{P\lambda}{P\lambda+k}}\sqrt{\frac{Pk}{P\lambda+k}\log(1/\beta)} \\
    &= \mathbb{E}[\regret] + \frac{P}{P\lambda+k}\sqrt{\lambda k \log(1/\beta)}
\end{align}

\subsubsection{Lower Bound}

Let $\mathcal{D}$ be a uniform distribution over vectors of welfare scores $w \in \lbrace 0, 1 \rbrace^P$ with $\lVert w \rVert_1 = k$. Because we are uniform, it makes no difference which set $\mathcal{I}$ the algorithm chooses beyond the value of $n$. The regret of the algorithm is exactly the number of non-0s in their final output times the maximum treatment effect, and so their optimal strategy is to first select every 0 revealed by $\mathcal{I}$. Then, if they've chosen fewer than $k'$ elements, they should sample the remainder from the unseen indices (and like before, it doesn't matter what strategy they use for this). Let $Z$ be the Hypergeometric random variable representing the number of 0s found in our initial sample. Then we have:
\begin{align}
    \mathbb{E}[\textsc{Regret}(A, w)] &= k - \mathbb{E}\left[ \min(Z, k') + \mathbb{I}[Z<k']\frac{(k'-Z)(k'-Z)}{P-n} \right]
\end{align}

Here we note that under the condition that $k' \leq k - (P - n) \Rightarrow \lambda n \geq P-n \Rightarrow n \geq \frac{P}{1+\lambda}$, the pigeonhole principle ensures that $\mathbb{P}(Z \geq k') = 1$, in which case our regret is exactly $k-k' = \lambda n$. Assuming otherwise, we have:
\begin{align}
    & \mathbb{E}[\textsc{Regret}(A, w)]\\
    &\geq k -  \min\left(k',\frac{nk}{P}\right) - \mathbb{E}\left[\frac{(k'-Z)(k'-Z)}{P-n} \right] \\
    &= k - \min\left(k',\frac{nk}{P}\right) - \frac{1}{P-n}\left( k'^2 - 2k'\frac{nk}{P} + \frac{nk(P-k)(P-n)}{P^3} + \frac{(nk)^2}{P^2} \right) \\
    &\geq k - \min\left(k',\frac{nk}{P}\right) - \frac{1}{P^2(P-n)}\left( P^2k'^2 - 2Pnkk' + n^2k^2 \right) - \frac{n}{4P} \\
    &= k - \min\left(k',\frac{nk}{P}\right) - \frac{1}{P^2(P-n)}\left( Pk' - nk \right)^2 - \frac{n}{4P} \\
    &= k - \min\left(k',\frac{nk}{P}\right) - \frac{(k' - \frac{nk}{P})^2}{P-n} - \frac{n}{4P} \\
    &= \max\left(\lambda n, k\left(1 - \frac{n}{P} \right) \right) - \frac{(k - n(\lambda + \frac{k}{P}))^2}{P-n} - \frac{n}{4P}
\end{align}

We will show that this function is concave over the intervals $n \in [0, \frac{Pk}{P\lambda+k}]$ and $[\frac{Pk}{P\lambda + k}, \frac{P}{1+\lambda}]$ and therefore minimized at one of those three extreme points.

Over both intervals, the max is always equal to a linear function of $n$, and so it suffices to prove that the quadratic term is concave. Let $C = \lambda + \frac{k}{P}$. Then we have:

\begin{align}
    \frac{-(k - nC)^2}{P-n}
    &= \frac{-(k + (P-n-P)C)^2}{P-n} \\
    &=  \frac{-((P-n)C + (k-PC))^2}{P-n} \\
    &= \frac{-((P-n)C + P\lambda)^2}{P-n} \\
    &= -(P-n)C^2 - 2CP\lambda - \frac{P^2\lambda^2}{P-n}
\end{align}
Taking second derivative with respect to $n$, we get:
\begin{align}
    - \frac{2P^2\lambda^2}{(P-n)^3}
\end{align}
which is always negative, showing that our lower bound is a concave function of $n$ over the given intervals, as desired. We now evaluate our objective function at each critical value, ignoring the effectively-constant term $\frac{n}{4P}$ for simplicity:
\begin{equation}
    \begin{cases}
        k\left(1 - \frac{k}{P} \right) & n=0 \\
        \frac{P\lambda k}{P\lambda+k} & n=\frac{Pk}{P\lambda + k} \\
        \frac{P\lambda}{1+\lambda} & n=\frac{P}{1+\lambda}
    \end{cases}
\end{equation}

From the fact that $k \leq P$, we can immediately exclude $n = \frac{P}{1+\lambda}$. Choosing $n=\frac{Pk}{k+P\lambda}$ is preferable to random allocation when:
\begin{align}
    \frac{P\lambda k}{k+P\lambda} 
    &\leq 
    k - \frac{k^2}{P} 
    \\
    P^2 \lambda k
    &\leq
    k^2P + kP^2\lambda - k^3 - k^2P\lambda
    \\
    \lambda 
    &\leq 
    \frac{k^2(P-k)}{k^2P}
    \\
    \lambda 
    &\leq 
    1 - \frac{k}{P}
\end{align}

We conclude that the worst-case regret of any fixed-budget algorithm for ILA with sampling must be at least:
\begin{equation}
    \begin{cases}
        k\left(1-\frac{k}{P}\right) & \lambda \geq 1 - \frac{k}{P} \\
        \frac{(P\lambda - \frac{1}{4}) k}{P\lambda + k} & \lambda < 1 - \frac{k}{P}
    \end{cases}
\end{equation}
as desired.

\subsection{Proof of Lemma \ref{lem:inequality-variance}}

Consider $\rho$, the vector of true probabilities, and let $\sigma^2 = \frac{1}{M} \sum_{j=1}^M \rho_j(1-\rho_j)$. Our goal is to upper-bound $\sigma^2$ in terms of the Gini coefficient $G$. We accomplish this by maximizing $G$ subject to constraints on $\bar\rho$ and $\sigma^2$. I.E. we want to maximize:

\begin{align}
    f(\rho) \coloneq M^2 \bar\rho \cdot G = \sum_{1 \leq i < j \leq M} (\rho_{s(j)} - \rho_{s(i)}) = \sum_{k=1}^M (2k -M - 1) \rho_{s(k)}
\end{align}
subject to the constraints that:
\begin{align}
    \left[\frac{1}{M}\sum_{i=1}^M \rho_{s(i)}(1-\rho_{s(i)})\right] - \sigma^2 = 0 \\
    \frac{1}{M} \sum_{i=1}^M \rho_{s(i)} - \bar\rho = 0
\end{align}

Here, we ignore the requirement that $\rho_i \in [0,1]$. This can only increase the value of our objective function and therefore does not impact the validity of our final bound. Moving on, we define the Lagrangian:
\begin{equation}
    \mathcal{L}(x, \lambda_1, \lambda_2) = \sum_{k=1}^M (2k -M - 1) \rho_{s(k)} 
    + \lambda_1 \left(\left[\frac{1}{M}\sum_{i=1}^M \rho_{s(i)}(1-\rho_{s(i)})\right] - \sigma^2 \right) 
    + \lambda_2 \left( \frac{1}{M} \sum_{i=1}^M \rho_{s(i)} - \bar\rho \right) 
\end{equation}
Taking partial derivatives, we get:
\begin{align}
    \frac{\partial \mathcal{L}}{\partial \rho_{s(k)}} &= 2k - M - 1 + \frac{\lambda_1}{M}(1 - 2\rho_{s(k)}) + \frac{\lambda_2}{M} \\
    \frac{\partial \mathcal{L}}{\partial \lambda_1} &= \left[\frac{1}{M}\sum_{i=1}^M \rho_{s(i)}(1-\rho_{s(i)})\right] - \sigma^2  \\
    \frac{\partial \mathcal{L}}{\partial \lambda_2} &=  \frac{1}{M} \sum_{i=1}^M \rho_{s(i)} - \bar\rho \\
\end{align}
Setting all of these equal to 0 and rearranging the first one, we get:
\begin{equation}
    \rho_{s(k)} = \frac{M}{2\lambda_1}\left(2k - M - 1 + \frac{\lambda_2}{M}\right) + \frac{1}{2} 
    = \frac{M\left(2k - M - 1\right) + \lambda_2 + \lambda_1}{2\lambda_1} 
\end{equation}
Plugging in the second constraint then gives us:
\begin{align}
    \bar\rho &= \frac{1}{2\lambda_1} \left(\lambda_2 + \lambda_1 +\sum_{k=1}^M (2k-M-1)\right) \\
    2\lambda_1 (\bar\rho - 1/2) &= \lambda_2
\end{align}
And plugging this back into the above gives us:
\begin{equation}
    \rho_{s(k)} = 
    \frac{M\left(2k - M - 1\right) + 2\lambda_1(\bar\rho-1/2) + \lambda_1}{2\lambda_1} = \frac{M(2k-M-1)}{2\lambda_1} + \bar\rho
\end{equation}
i.e. the optimal unit profile vector is evenly spaced. Finally, plugging this in to the first constraint gives us:
\begin{align}
    \sigma^2 
    &= \frac{1}{M} \sum_{k=1}^M \left( \frac{M(2k-M-1)}{2\lambda_1} + \bar\rho\right) \left( 1 - \frac{M(2k-M-1)}{2\lambda_1} - \bar\rho\right) \\
    &= \bar\rho - \frac{1}{M}\sum_{k=1}^M \left(\frac{M(2k-M-1)}{2\lambda_1} + \bar\rho\right)^2 \\
    &= \bar\rho(1 - \bar\rho) - \frac{1}{M} \sum_{k=1}^M \left(\frac{M(2k-M-1)}{2\lambda_1}\right)^2 \\
    \Rightarrow 4\lambda_1^2(\sigma^2 - \bar\rho(1-\bar\rho))
    &= \frac{-1}{M}\sum_{k=1}^M \left(M(2k-M-1)\right)^2 \\
    &= \frac{-1}{3}M^2(M^2-1) \\
    \Rightarrow \lambda_1^2
    &= \frac{M^2(M^2-1)}{12(\bar\rho(1-\bar\rho) - \sigma^2)} \\
    \Rightarrow \lambda_1 &= 
    M \sqrt{\frac{M^2-1}{12(\bar\rho(1-\bar\rho) - \sigma^2)}}
\end{align}

This then implies that:
\begin{equation}
    \rho_{s(k)} = \bar\rho + \frac{\sqrt{3(\bar\rho(1-\bar\rho) - \sigma^2)}(2k-M-1)}{\sqrt{M^2-1}}
\end{equation}

Finally, we plug this into our objective function to get:

\begin{align}
    f(\rho) \coloneq M^2 \bar\rho \cdot G 
    &= \sum_{k=1}^M (2k -M - 1) \rho_{s(k)} \\
    &= \sum_{k=1}^M (2k -M - 1) \left( \bar\rho + \frac{\sqrt{3(\bar\rho(1-\bar\rho) - \sigma^2)}(2k-M-1)}{\sqrt{M^2-1}} \right) \\ 
    &= \sum_{k=1}^M (2k -M - 1)^2 \left(\frac{\sqrt{3(\bar\rho(1-\bar\rho) - \sigma^2)}}{\sqrt{M^2-1}} \right) \\
    &= M\sqrt{\frac{(M^2-1)(\bar\rho(1 - \bar\rho) - \sigma^2)}{3}}
\end{align}

Implying that:

\begin{equation}
    G = \sqrt{\frac{(M^2-1)(\bar\rho(1 - \bar\rho) - \sigma^2)}{3M^2\bar\rho^2}} \leq \sqrt{\frac{1 - \bar\rho}{3} - \frac{\sigma^2}{3\bar\rho^2} }
\end{equation}

and therefore that:

\begin{equation}
    \sigma^2 \leq \bar\rho(1-\bar\rho) - 3\bar\rho^2 \cdot G^2
\end{equation}

as desired. \qed

\subsection{Proof of Theorem \ref{thm:ula-sample}}

We begin by considering a simplification of the sampling problem with fixed $k$ and $n$ before proceeding to consider the fixed-budget variant.

Fix $k$ and $n$. By Lemma \ref{lem:gaussian-mechanism} and the central limit theorem for hypergeometric random variables~\cite{madow1948limiting}, we know that each component $\widetilde\rho_j$ of our private estimated unit profile vector is approximately Gaussian with variance:
\begin{equation}
    \sigma_j^2 = \frac{M^2}{2\psi n^2} + \frac{M(N-n/M)\rho_j(1-\rho_j)}{nN}
\end{equation}
Our goal is to bound the $\ell_p$ norm of this vector for $p \in \lbrace1, 2, \infty \rbrace$ so that we can apply Lemma \ref{lem:ula-helper}. To start, we have that:
\begin{equation}
    \sigma_* \coloneq\max_j \sigma_j^2 \leq \frac{M^2}{2\psi n^2} + \frac{M(N-n/M)}{4nN} = \frac{M^2}{2\psi n^2} + \frac{P-n}{4nN}
\end{equation}
and by Lemma \ref{lem:inequality-variance}, we know that:
\begin{equation}
    \sigma^2 \coloneq\frac{1}{M} \sum_{j=1}^M \sigma_j^2 \leq \frac{M^2}{2\psi n^2} + \frac{(P-n)(\bar\rho_M(1-\bar\rho_M)-3\bar\rho_M^2G^2)}{nN}
\end{equation}
Finally, Jensen's inequality then gives us that:
\begin{equation}
    \bar\sigma \coloneq \frac{1}{M} \sum_{j=1}^M \sigma_j \leq \sigma
\end{equation}
Applying Lemma \ref{lem:ula-helper}, our combined error from privacy and statistical estimation becomes:
\begin{equation}
    \begin{cases}
        \sqrt{\frac{M^2}{2\psi n^2} + \frac{P-n}{4nN}} \cdot 2\min(k, P-k)\sqrt{2\log(2M)} & p=1, q=\infty \\
        \sqrt{\frac{M^2}{2\psi n^2} + \frac{(P-n)(\bar\rho_M(1-\bar\rho_M)-3\bar\rho_M^2G^2)}{nN}} \cdot \sqrt{2\min(k, P-k)P} & p=2, q=2 \\
        \sqrt{\frac{M^2}{2\psi n^2} + \frac{(P-n)(\bar\rho_M(1-\bar\rho_M)-3\bar\rho_M^2G^2)}{nN}} \cdot P\sqrt{\frac{2}{\pi}} & p=\infty, q=1
    \end{cases}
\end{equation}

\subsection{Proof of Corollary \ref{cor:ula-sample-budget}}

We now consider optimizing regret subject to budgetary constraints. Let $k' = k - n\lambda \geq 0$ denote the actual number of individuals we treat. The simplest approach would be to imagine that we first decide on a $k$-allocation and then arbitrarily remove $n\lambda$ individuals, which can increase our normalized regret by at most $n\lambda$. If we instead compute a $k$-allocation and then randomly allocate aid to a $\frac{k-\lambda n}{k}$ portion of the individuals in each treated unit, then we obtain a tighter bound. Let $\tilde V$ denote the value of the private $k$-allocation, $\tilde V_\lambda$ be the value of the subsampled allocation, and $V^*$ the value of the optimal allocation. Then we have:
\begin{align}
    \tilde V_\lambda &= \left(1 - \frac{\lambda n}{k} \right) \tilde V \\
    V^* - \tilde V_\lambda &= V^* - \left(1 - \frac{\lambda n}{k} \right) \tilde V
\end{align}
So, our regret is maximized when $\tilde V$ reaches its minimum value, implying that our total expected regret is at most:
\begin{equation}
    \overline\regret(\widetilde\ULA_\lambda, w; k) \leq \overline\regret{\widetilde \ULA, w; k)} + \lambda n (1-\bar\rho_K)
\end{equation}

This is strictly better than the na\"ive upper bound of $\lambda n$ whenever $\bar\rho_K > 0$, although we note that $\bar\rho_K$ is not generally available prior to deciding on the value of $n$. We will therefore ignore all terms depending on $\bar\rho_M$, $\bar\rho_K$, and $G$ when optimizing parameters.

From here, we consider possible extreme point solutions. The choice of $n=0$ essentially reduces to random allocation in expectation, which we have analyzed elsewhere. Meanwhile, the choice of $n=P$ allows us to reduce to the Direct setting, but with an additional linear error term of $P\lambda(1-\bar\rho_K)$. Otherwise, we have the following simplified upper-bound:

\begin{equation}
    \lambda n  +
    C\sqrt{\frac{M^2}{2\psi n^2} +  \frac{P-n}{4nN}} \leq \lambda n + C \left(\sqrt{\frac{M^2}{2\psi n^2}} + \sqrt{\frac{P-n}{4nN}} \right)
\end{equation}
where
\begin{equation}
    C = 
    \min\begin{cases}
        2\min(k, P-k)\sqrt{2\log(2M)}  \\
        \sqrt{2\min(k, P-k)P} \\
        P\sqrt{\frac{2}{\pi}}
    \end{cases}
\end{equation}

We optimize each term independently. If the first term is larger, then taking derivative and setting equal to 0, we get:

\begin{align}
    \lambda + C\left( -\frac{\sqrt{M^2}}{n^2\sqrt{2\psi}} \right) &= 0 \\
    \lambda n^2 \sqrt{2\psi} &= CM \\
    n = \sqrt{\frac{CM}{\lambda \sqrt{2\psi}}}
\end{align}

In which case the objective becomes:
\begin{equation}
    \lambda n + \sqrt{\frac{\lambda CM}{\sqrt{2\psi}}}
\end{equation}

If the second term is larger, we additionally upper bound $\frac{P-n}{N} \leq \frac{P}{N} = M$, and get:

\begin{align}
    \lambda + \frac{C}{2}\sqrt{\frac{4n}{M}}\left( \frac{-M}{4n^2} \right) &= 0 \\
    4 \lambda  n^2 &= C\sqrt{Mn} \\
    16 \lambda^2 n^3 &= C^2 M \\
    n &=  \left(\frac{C^2M}{16\lambda^2} \right)^{1/3}
\end{align}

Plugging this into our objective gives us:

\begin{equation}
    \lambda n + \frac{C}{2} \sqrt{\frac{4^{2/3}M^{2/3}\lambda^{2/3}}{C^{2/3}}} = \left( \frac{C^{2}M\lambda}{2} \right)^{1/3}
\end{equation}

Finally, choosing $n = \max \left(\sqrt{\frac{CM}{\lambda \sqrt{2\psi}}}  , \left(\frac{C^2M}{16\lambda^2} \right)^{1/3}\right)$ guarantees a final regret bound of:
\begin{equation}
    \max \left(\sqrt{\frac{CM\lambda}{\sqrt{2\psi}}}  , \left(\frac{C^2M\lambda}{16} \right)^{1/3}\right) + \sqrt{\frac{CM \lambda }{\sqrt{2\psi}}} + \left( \frac{C^{2}M\lambda}{2} \right)^{1/3} 
    \leq 2 \sqrt{\frac{CM\lambda}{\sqrt{2\psi}}} + \frac{2^{2/3}}{4} \left(C^2M\lambda \right)^{1/3}
\end{equation}

Since $C = O(P)$, it follows that $n = O(P^{2/3})$, and so in fact $\frac{P-n}{N} \sim M$ and we haven't lost anything from the additional upper bound.

\section{Missing Proofs from \Cref{sec:learning}}

\subsection{Proof of Lemma \ref{lem:ila-sco}}

For each $i$, let $y_i$ denote the realized binary label, and $\hat \eta_i$ denote the probability predicted by our model. By applying Lemma \ref{lem:ila-helper}, we obtain a normalized regret bound with respect to optimal $k$-allocation in hindsight of:
\begin{equation}
    k-k' + k' \bar \eta_{k'} + \sum_{i \in \mathcal{I}} |y_i - \hat y_i| \leq k-k' + k' \bar \eta_{k'} + \sum_{i=1}^P |y_i - \hat y_i|
\end{equation}

Since $\mathbb{E}[|y_i - \hat y_i|]^2 \leq \mathbb{E}[(y_i - \hat y_i)^2] = \alpha$, it follows that $\mathbb{E}[|y_i - \hat y_i|] \leq \sqrt{\alpha}$, and so in expectation our regret against the Bayes' optimal ILA is bounded by $k-k' + k' \bar \eta_{k'} +P\sqrt{\alpha}$.

Since this upper-bound comes from a cap on the population error, we can immediately derive a high-probability version by constructing a high-probability bound on the quantity:
\begin{equation}
    \sum_{i=1}^P (y_i - \hat y_i)^2
\end{equation}
since these are bounded independent random variables in $[0,1]$, Hoeffding gives us that:
\begin{equation}
    \mathbb{P} \left( \sum_{i=1}^P (y_i - \hat y_i)^2 > P\alpha + t \right) \leq \exp\left( \frac{-2t^2}{P} \right)
\end{equation}
And so, with probability at least $1-\beta$, we have:
\begin{align}
    \sum_{i=1}^P (y_i - \hat y_i)^2 
    & \leq 
    P\alpha + \sqrt{\frac{P\log(1/\beta)}{2}}
\end{align}
Combined with the above, this gives a total regret of:
\begin{equation}
    k-k' + k'\bar \eta_{k'} + P\sqrt{\alpha + \sqrt{\frac{\log(1/\beta)}{2P}}} \leq k-k' + k'\bar \eta_{k'} +P\sqrt{\alpha} + P^{3/4}\sqrt{\log(1/\beta)/2}
\end{equation}

\subsection{Proof of Theorem \ref{thm:ila-sco-private}}

The theorem follows directly by combining Lemma \ref{lem:ila-sco} with the utility analysis of \autoref{alg:ila-direct}.

\subsection{Proof of Corollary \ref{cor:ila-sco-budget}}

 When $\widetilde\ILA$ is instantiated with Algorithm 2 of \cite{feldman2020private}, its asymptotic normalized regret grows like:
    \begin{equation}
        n\lambda + P\sqrt{\sigma^2 + E^* + 10LD \left( \frac{1}{\sqrt{n}} + \frac{\sqrt{p}}{n\sqrt{2\psi}} \right)}
    \end{equation}

\subsubsection{Generic Bounds}

To derive a bound that doesn't depend on using the unknown quantity $\mathcal{L}(\theta^*)$ when choosing parameters, we use that:

\begin{equation}
    \sqrt{\sigma^2 + E^* + 10LD \left( \frac{1}{\sqrt{n}} + \frac{\sqrt{p}}{n\sqrt{2\psi}} \right)} \leq \sqrt{\mathcal{L}(\theta^*)} + \sqrt{10LD \left( \frac{1}{\sqrt{n}} + \frac{\sqrt{p}}{n\sqrt{2\psi}} \right)}
\end{equation}

Supposing that $n$ is small so the privacy term dominates, we get:
\begin{equation}
    n\lambda + P\sqrt{\frac{10LD\sqrt{p}}{n\sqrt{2\psi}}} \eqqcolon n\lambda + \frac{A_1}{\sqrt{n}}
\end{equation}

Choosing $n = \left( \frac{A_1}{\lambda} \right)^{2/3}$ yields a final regret of:
\begin{equation}
    P\sqrt{\mathcal{L}(\theta^*)} + 2A_1^{2/3}\lambda^{1/3} = P\sqrt{\mathcal{L}(\theta^*)} +  2\lambda^{1/3} \left(\frac{P\sqrt{10LD\sqrt{p}}}{(2\psi)^{1/4}} \right)^{2/3} = O\left(P\sqrt{\mathcal{L}(\theta^*)} +  (LD\lambda/\sqrt{\psi})^{1/3}P^{2/3}p^{1/6} \right)
\end{equation}

And if instead $n$ is large, we get:

\begin{equation}
 n\lambda + P\sqrt{\frac{10LD}{\sqrt{n}}} \eqqcolon n\lambda + A_2 n^{-1/4}    
\end{equation}

Choosing $n = (A_2/\lambda)^{4/5}$ yields:
\begin{equation}
    P\sqrt{\mathcal{L}(\theta^*)} + 2A_2^{4/5}\lambda^{1/5} = P\sqrt{\mathcal{L}(\theta^*)} +  2\lambda^{1/5} \left( P\sqrt{10LD} \right)^{4/5} = O\left(P\sqrt{\mathcal{L}(\theta^*)} +  \lambda^{1/5}P^{4/5}(LD)^{2/5} \right)
\end{equation}

So, by choosing $n = \max\left( (A_1/\lambda)^{2/3}, (A_2/\lambda)^{4/5} \right)$, we can guarantee a worst-case regret bound of:
\begin{equation}
    P\sqrt{\mathcal{L}(\theta^*)} + 2\lambda^{1/3} \left(\frac{P\sqrt{10LD\sqrt{p}}}{(2\psi)^{1/4}} \right)^{2/3} + 2\lambda^{1/5} \left( P\sqrt{10LD} \right)^{4/5}
\end{equation}

\subsubsection{Instance-Dependent Bounds}

If we instead have access to a lower bound $\mathcal{L}(\theta^*) \geq L^\downarrow > 0$, then we can upper-bound our regret using the Taylor expansion:

\begin{equation}
    \lambda n + P\sqrt{\sigma^2 + E^* + 10LD \left( \frac{1}{\sqrt{n}} + \frac{\sqrt{p}}{n\sqrt{2\psi}} \right)} \leq \lambda n + P\sqrt{\mathcal{L}(\theta^*)} + \frac{10PLD \left( \frac{1}{\sqrt{n}} + \frac{\sqrt{p}}{n\sqrt{2\psi}} \right)}{2\sqrt{L^\downarrow}}
\end{equation}

Suppose that $n$ is small so that the privacy term dominates. Then we get:
\begin{equation}
    \lambda n + \frac{5PLD\sqrt{p}}{n\sqrt{2\psi L^\downarrow}} \eqqcolon \lambda n + \frac{A_3}{n}
\end{equation}
The right-hand side is optimized when $n = \sqrt{\frac{A_3}{\lambda}}$. Plugging into the objective function then gives:
\begin{equation}
    P\sqrt{\mathcal{L}(\theta^*)} + 2\sqrt{A_3\lambda} = P\sqrt{\mathcal{L}(\theta^*)} + 2\sqrt{\frac{5PLD\lambda\sqrt{p}}{\sqrt{2\psi L^\downarrow}}} = O\left(P\sqrt{\mathcal{L}(\theta^*)} + \frac{(PLD\lambda)^{1/2}p^{1/4}}{\psi^{1/4}(L^\downarrow)^{1/4}} \right)
\end{equation}
If instead $n$ is large so the statistical term dominates, then we get:
\begin{equation}
    \lambda n + P\sqrt{\mathcal{L}(\theta^*)} + \frac{5PLD}{\sqrt{n L^\downarrow}} \eqqcolon\lambda n + P\sqrt{\mathcal{L}(\theta^*)} + \frac{A_4}{\sqrt{n}}
\end{equation}
The right-hand-side is optimized when $n = \left( \frac{A_4}{\lambda} \right)^{2/3}$, and plugging into our objective function gives us:
\begin{equation}
    P\sqrt{\mathcal{L}(\theta^*)} + 2(A_4\lambda)^{2/3} = P\sqrt{\mathcal{L}(\theta^*)} +  2\left( \frac{5PLD\lambda}{\sqrt{L^\downarrow}} \right)^{2/3} = O\left(P\sqrt{\mathcal{L}(\theta^*)} + (PLD\lambda)^{2/3} (L^\downarrow)^{-1/3} \right)
\end{equation}

So, choosing $n = \max\left( \sqrt{A_3/\lambda}, (A_4/\lambda)^{2/3} \right)$ gives a combined regret bound of:
\begin{equation}
    P\sqrt{\mathcal{L}(\theta^*)} + 2\sqrt{\frac{5PLD\lambda\sqrt{p}}{ \sqrt{2\psi L^\downarrow}}} + 2\left( \frac{5PLD\lambda}{\sqrt{L^\downarrow}} \right)^{2/3}
\end{equation}

\subsection{Proof of Lemma \ref{lem:private-membership}}

Our approach is presented in \autoref{alg:ula-private-membership}. The intuition behind the algorithm is that non-private ULA is equivalent to ILA when each individual's welfare score is replaced with the $\rho_j$ score of their unit. So, given a $\psi_1$-zCDP estimate of the unit profile vector, we can repurpose \autoref{alg:ila-direct} with privacy paramter $\psi_2$ to compute a Joint DP allocation. The final allocation then satisfies Joint $(\psi_1 + \psi_2)$-zCDP by composition~\cite{bun2016concentrated} and Lemma \ref{lem:billboard}.

Two small adjustments must be made when estimating the unit profile vector. First, we must redefine our notion of adjacency to include unit membership, which increases the replace-one sensitivity of $\rho$ by a factor of $\sqrt{2}$. Secondly, because the sensitivity of $\rho_j$ scales with $|U_j|^{-1}$, which is now private, getting the additive noise scale as tight as possible requires some subtlety~\cite{kulesza2023mean}. For ease of exposition, we sidestep this point by assuming that the population of each unit is lower-bounded by $N$, and that this fact is publicly known. 

With those adjustments made, our regret decomposes into the baseline regret of ULA, the regret from estimating $\rho$, and the regret from \autoref{alg:ila-direct}. Combined, we get that in expectation, $\overline\regret(\widetilde \ULA, w; k)$ is at most:
\begin{equation}
    k\bar\rho_K + (1+o(1)) \left[ \frac{3\log^{3/2}P}{\pi\sqrt{\psi_2}} \right] + \frac{\sqrt{2}}{N\sqrt{\psi_1}} \cdot C
\end{equation}
Choosing $\frac{\psi_2}{\psi_1} = \left( \frac{3\log^{3/2}P}{\pi} \cdot \frac{N}{C\sqrt{2}} \right)^{2/3}$ then yields a combined regret bound of:
\begin{equation}
    k\bar\rho_K + \frac{1}{\sqrt{\psi}} \left[ \left(\frac{3}{\pi} \right)^{2/3} \log P + \left( \frac{C\sqrt{2}}{N} \right)^{2/3} \right]^{3/2} = k\bar\rho_K + O\left( \frac{C}{N\sqrt{\psi}} \right)
\end{equation}

\begin{remark}
    Under the condition that $|\widetilde\rho_i - \widetilde\rho_j| > \frac{2}{k\pi\sqrt{\psi_2}}$ for all $i \neq j$, then \autoref{alg:ula-private-membership} satisfies the property (shared by \autoref{alg:ula-direct}) that if any individual from a unit with estimated profile $\widetilde\rho_j$ receives treatment, then all individuals from units with estimated profile $\widetilde\rho_i < \widetilde\rho_j$ must receive treatment as well. When units are tied or nearly tied, then it is possible that \autoref{alg:ula-private-membership} will allocate aid to only a fraction of individuals across multiple units.
\end{remark}

\subsection{Proof of Lemma \ref{lem:ula-sco-baseline}}

From the proof of Lemma \ref{lem:ula-helper}, we know that the regret of ULA with estimated welfare scores is at most:
\begin{align}
    k\bar\rho_K +\sum_{j=1}^M (x_j^* - \tilde x_j^*) (\widetilde\rho_j - \rho_j)
    &\leq k\bar\rho_K + \sum_{j=1}^M |x_j^* - \tilde x_j^*| |\widetilde\rho_j - \rho_j| \\
    &\leq k\bar\rho_K + \sum_{j=1}^M N_j |\widetilde\rho_j - \rho_j|
\end{align}

From here, we have that:

\begin{align}
\mathbb{E}\left[\sum_{j=1}^M N_j |\widetilde\rho_j - \rho_j| \right]
&=
    \mathbb{E}\left[ \sum_{k=1}^{|\mathscr{P}|} \left|\sum_{i \in U_k} y_i - \hat y_i \right| \right] \\
    &= \mathbb{E}\left[ \sum_{k=1}^{|\mathscr{P}|} \left|\sum_{i \in U_k} 
        (y_i - \mathbb{E}[y_i|x_i]) - 
        (\hat y_i - \mathbb{E}[ y_i|x_i]) \right| \right] \\
    &\leq \mathbb{E}\left[ \sum_{k=1}^{|\mathscr{P}|} 
        \left|\sum_{i \in U_k} (y_i - \mathbb{E}[y_i|x_i]) \right| +
    \sum_{k=1}^{|\mathscr{P}|}
        \left| \sum_{i \in U_k} (\hat y_i - f(x_i)) \right| \right] \\
    &\leq \sqrt{P|\mathscr{P}| \sigma^2} + P\cdot\mathbb{E}\left[|\hat y - \mathbb{E}[y\mid x]| \right] \\
    &\leq \sqrt{P |\mathscr{P}|\sigma^2} + P\sqrt{\alpha - \sigma^2}
\end{align}

As desired.

\subsection{Proof of Theorem \ref{thm:ula-sco-private}}

The theorem follows directly from combining the regret terms from Lemma \ref{lem:ula-sco-baseline} and Lemma \ref{lem:private-membership}.

\subsection{Proof of Corollary \ref{cor:ula-sco-budget}}

The logic of the proof is essentially identical to the proof of Corollary \ref{cor:ila-sco-budget}, except that we replace all instances of $\mathcal{L}(\theta^*)$ with $E^*$ and replace the lower bound $L^\downarrow$ with $0<E^\downarrow \leq E^*$.


\end{document}